%% file: functionality.tex
\documentclass{llncs}
 \usepackage{amsmath}
 \usepackage{amssymb}
 \usepackage{stmaryrd}
 \usepackage{xspace}
 \usepackage{times}

\usepackage{xcolor,pgf,tikz,pgflibraryarrows,pgffor,pgflibrarysnakes}
\usetikzlibrary{snakes}

\usepackage{gastex}

\input{macro}

  \title{On Functionality of Visibly Pushdown Transducers}

\author{Emmanuel Filiot$^\dag$ $\qquad$ Jean-Fran\c{c}ois Raskin$^\dag$ $\qquad$\ Pierre-Alain Reynier$^*$\\
  Fr\'ed\'eric Servais$^\dag$  $\qquad$  Jean-Marc Talbot$^*$ }

\date{} 
\institute{$^\dag$ Universit\'e Libre de Bruxelles \\ $^*$ Universit\'e de Provence}

\begin{document}

\maketitle

\begin{abstract}
    Visibly pushdown transducers 
    form a subclass of pushdown transducers that (strictly) extends
    finite state transducers with a stack. Like visibly pushdown automata, the input symbols
    determine the stack operations. In this paper, we prove that functionality
    is decidable in \pspace for visibly pushdown transducers. The proof is done via a pumping
    argument: if a word with two outputs has a sufficiently large
    nesting depth, there
    exists a nested word with two outputs whose nesting depth is
    strictly smaller. The proof uses technics of word combinatorics.
    As a consequence of decidability of functionality, we also show that
    equivalence of functional visibly pushdown transducers is \exptimecomplete. 
\end{abstract}

\input{introduction}

\input{definitions}

\input{pumping}

\input{equations}
\input{pspace}

\bibliographystyle{abbrv}
\bibliography{biblio}

\end{document}

%% file: macro.tex
\newcommand{\vpa}{\ensuremath{\mathsf{VPA}}\xspace}

\newcommand{\vpt}{\ensuremath{\mathsf{VPT}}\xspace}
\newcommand{\vpts}{\ensuremath{\mathsf{VPTs}}\xspace}

\newcommand{\wn}{\ensuremath{\mathsf{wn}}\xspace}

\newcommand{\vpl}{\ensuremath{\mathsf{VPL}}\xspace}

\newcommand{\fsts}{\ensuremath{\mathsf{FSTs}}\xspace}
\newcommand{\fst}{\ensuremath{\mathsf{FST}}\xspace}

\newcommand{\nlogspace}{{\sc NLogSpace}\xspace}
\newcommand{\conlogspace}{{\sc Co-NLogSpace}\xspace}
\newcommand{\exptime}{{\sc ExpTime}\xspace}
\newcommand{\ptime}{{\sc PTime}\xspace}
\newcommand{\pspace}{{\sc PSpace}\xspace}
\newcommand{\exptimecomplete}{{\sc ExpTime-c}\xspace}

\newcommand{\Sys}{\mathcal{S}}

\newcommand{\img}{\textit{CoDom}} 
\newcommand{\dom}{\textit{Dom}} 

\newcommand\inter[1]{\llbracket #1 \rrbracket}

\newcommand{\ol}{\overline}

%% file: introduction.tex
\section{Introduction}

In \cite{RM04}, it has been shown that visibly pushdown languages (\vpl) form a robust subclass of context-free languages. 
This class strictly extends the class of regular languages and still enjoys strong properties: 
closure under all Boolean operators and decidability of emptiness, universality, inclusion and equivalence.
On the contrary, context-free languages are not closed under complement nor under intersection, moreover universality, inclusion 
and equivalence are all undecidable.

\textit{Visibly pushdown automata} (\vpa), that characterize \vpl, are obtained as a restriction of pushdown automata. 
In these automata the input symbol determines the stack operation. The input alphabet is partitioned into
call, return and internal symbols: if a call is read, the automaton must push a symbol on the stack; if it reads a return, 
it must pop a symbol; and while reading an internal symbol, it can not touch, not even read, the stack. 
\textit{Visibly pushdown transducers} have been introduced in \cite{RS08}. They
form a subclass of pushdown transducers, and are obtained by adding output to \vpa: each time the \vpa reads an input symbol 
it also outputs a letter. They allow for
$\epsilon$-transitions that can produce outputs. In this paper, we consider visibly pushdown transducers
where this operation is not allowed. Moreover, each transition can
output not only a single letter but a word, and no visibly restriction
is imposed on this output word. Therefore in the sequel we call the transducers of \cite{RS08} 
$\epsilon$-\vpts, and \vpts will denote the visibly pushdown
transducers considered here.

Consider the \vpt $T$ of Figure~\ref{fig:vpt}. Call (resp. return) symbols are
denoted by $c$ (resp. $r$). The domain of $T$ is $\dom(T) = \{
c_1(c_2)^nc_3r_3(r_2)^nr_1\ |\ n\in\mathbb{N}\}$. For each word of
$\dom(T)$, there are two accepting runs, corresponding respectively to
the upper and lower part of $T$. For instance, when reading $c_1$, it
pushes $\gamma_1$ and produces either $d$ (upper part) or $dfc$ (lower
part). By following the upper part (resp. lower part), it produces words of the form
$dfcab(cabcab)^ngh$ (resp. $dfc(abc)^nab(cab)^ngh$). Therefore $T$ is functional.

\input{figVPT}

In this paper, we prove that the problem of determining if a
\vpt transduction is functional is decidable. In particular, our algorithm is
in \pspace. Deciding functionality is one of the main
problem in transduction theory as it makes deciding equivalence of functional
transducers possible. Both problems are undecidable for pushdown
transductions. Our proof relies on a pumping argument: if a word is
long enough and has two outputs, we show that there is a strictly shorter word with
two outputs. We use technics of word combinatorics and in particular,
a strong result proved in \cite{Hakala99}. As a consequence, we show
that the equivalence problem for \vpts is \exptimecomplete.

\paragraph{Related Work}
$\epsilon$-\vpts have been introduced in \cite{RS08}. In contrast
to \vpts, they allow for $\epsilon$-transitions that produce outputs,
so that an arbitrary number of symbols can be inserted. Moreover, 
each transition of a \vpt can output a word while each transition of
an $\epsilon$-\vpt can output a single letter only. The \vpts we
consider here are strictly less expressive than $\epsilon$-\vpts, but
functionality and equivalence of functional transducers are decidable, which is not the case
for $\epsilon$-\vpts.

The functionality problem for finite state
transducers has been extensively studied. The first proof of
decidability was given by Sch{\"u}tzenberger in \cite{Schutz75}, and later in
\cite{JCSS::BlattnerH1977}. As the proof we give here, the proof of
Sch{\"u}tzenberger relies on a pumping lemma for functionality.
The first \ptime upper bound has
been proved in \cite{GurIba83}, and an efficient procedure has been
given in \cite{BealCPS03}.

Deciding equivalence of deterministic (and therefore functional) \vpts
is in \ptime \cite{SLLN09}. However, functional \vpts are strictly
more expressive than deterministic \vpts. In particular, 
non-determinism is often needed to
model functional transformations whose current production depends on some
input which may be arbitrary far away from the current input. For instance,
the transformation that swaps the first and the last input symbols is
functional but non-determinism is needed to guess the last input.

Ordered trees over an arbitrary finite alphabet $\Sigma$ can be naturally
represented by well nested words over the structured alphabet
$\Sigma\times\{c\}\cup \Sigma\times\{r\}$. As \vpts can express
transductions from well words to well nested words, they
are therefore well-suited to model tree tranformations. We distinguish 
\textit{ranked trees} from \textit{unranked trees}, whose nodes may have an arbitrary number of
ordered children. Ranked tree transducers have received a lot of
attention. Most notably, \textit{tree transducers} \cite{tata2007} and
\textit{macro tree transducers} \cite{Eng85} have been proposed and studied.
They are incomparable to \vpts however, as they allow for copy,
which is not the case of \vpts, but cannot define 
any context-free language as codomain, what \vpts can do.
Functionality is known to be decidable in \ptime for tree transducers
\cite{TCS::Seidl1992}. More generally, finite-valuedness (and
equivalence) of tree transducers is decidable \cite{MST::Seidl1994}.
There have been several attempts to generalize ranked tree transducers
to unranked tree transducers \cite{Maneth:2000:SDT,PerSei04}. As mentioned in
\cite{Eng09}, it is an important problem to decide
equivalence for unranked tree transformation formalisms. However, 
there is no obvious generalization of known results for ranked trees
to unranked trees, as unranked tree transformations have to
support concatenation of tree sequences, making usual binary encodings
of unranked trees badly suited. Considering classical ranked tree
transducers, their ability to copy subtrees is the main concern when dealing 
with functionality. However for \vpts, it is more 
their ability to concatenate sequences of trees which makes this problem
difficult, and which in a way led us to word combinatorics. To the best of our knowledge, \vpts consist in the first
(non-deterministic) model of  unranked tree transformations for which
functionality and equivalence of functional transformations is
decidable.

\paragraph{Organization of the paper} In Section \ref{sec:def}, we
define visibly pushdown transducers as a
extension of visibly pushdown automata. In Section \ref{sec:wordeq}, we recall
some notion of word combinatorics. In Section \ref{sec:pumpingfun}, we give a
reduction of functionality 
to a system of word equations. In Section \ref{sec:equations}, we prove a pumping lemma
that preserves non-functionality. Finally, we give a \pspace algorithm for functionality is  Section 
\ref{sec:pspace} and prove the \exptime completeness of
equivalence.

%% file: figVPT.tex
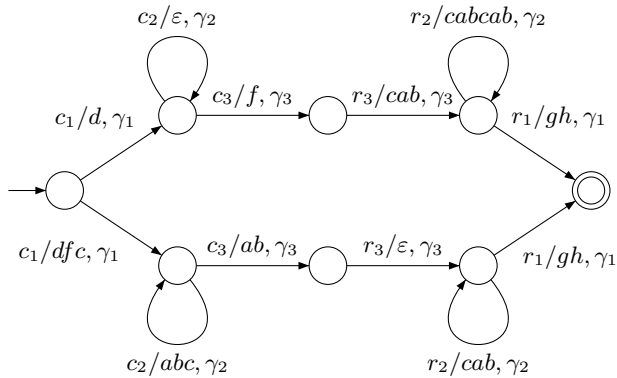
\begin{figure}[!ht]
\centering
\begin{picture}(60,50)(-5,-10)

\gasset{Nw=5,Nh=5,Nmr=5,NLdist=5,NLangle=180}

\node[Nmarks=i,NLangle=135](i)(-5,15){}
\node[NLangle=90](q1)(10,5){}
\node[NLangle=90](q2)(50,5){}
\node[NLangle=90](q3)(30,5){}
\node[NLangle=90](p3)(30,25){}
\node[NLangle=270](p1)(10,25){}
\node[NLangle=270](p2)(50,25){}
\node[NLangle=0,Nmarks=r](f)(65,15){}

\drawedge[ELside=r,ELpos=30](i,q1){$c_1 / dfc, \gamma_1$ }
\drawedge(q1,q3){$c_3 / ab, \gamma_3$}
\drawedge(q3,q2){$r_3 / \varepsilon, \gamma_3$}
\drawedge(i,p1){$c_1 / d, \gamma_1$ }
\drawedge(p1,p3){$c_3 / f, \gamma_3$}
\drawedge(p3,p2){$r_3 / cab, \gamma_3$}
\drawedge(p2,f){$r_1 / gh, \gamma_1$}
\drawedge[ELside=r,ELpos=60](q2,f){$r_1 / gh, \gamma_1$}

\drawloop(p1){$c_2 / \varepsilon, \gamma_2$}
\drawloop(p2){$r_2 / cabcab, \gamma_2$}
\drawloop[loopangle=270](q1){$c_2 / abc, \gamma_2$}
\drawloop[loopangle=270](q2){$r_2 / cab, \gamma_2$}

\end{picture}
\caption{A functional \vpt on $\Sigma_c=\{c_1,c_2,c_3\}$ and $\Sigma_r=\{r_1,r_2,r_3\}$.}\label{fig:vpt}
\end{figure}

%% file: definitions.tex
\section{Visibly Pushdown Transducers}\label{sec:def}

Let $\Sigma$ be a finite alphabet partitioned into two disjoint sets
$\Sigma_c$ and $\Sigma_r$ denoting respectively the \textit{call}
and \textit{return} alphabets\footnote{In contrast to \cite{RM04}, we do
  not consider \textit{internal} symbols $i$, 
  as they can be simulated by a (unique) call $c_i$ followed
  by a (unique) return $r_i$}. We denote by $\Sigma^*$ the set of
words over $\Sigma$ and by $\epsilon$ the empty word. The length of a word $u$ is denoted 
by $|u|$. The set of \textit{well nested} words $\Sigma^*_\wn$ is the smallest subset
of $\Sigma^*$ such that $\epsilon\in\Sigma^*_\wn$ and for all $c\in\Sigma^c$, all
$r\in \Sigma^r$, all $u,v\in \Sigma^*_\wn$, $cur\in \Sigma^*_\wn$ and
$uv\in\Sigma^*_\wn$. The \textit{height} of
a well nested word is inductively defined
by $h(\epsilon) = 0$, $h(cur) = 1+h(u)$, and $h(uv) = \text{max}(h(u),h(v))$.

\paragraph{Visibly Pushdown Languages} 
A {\em visibly pushdown automaton}
(\vpa)~\cite{RM04} on finite words over $\Sigma$
is a tuple $A=(Q,I, F, \Gamma, \delta)$ where 
$Q$ is a finite set of states, $I\subseteq Q$, respectively
$F\subseteq Q$, the set of initial states, respectively final
states, $\Gamma$ the (finite) stack alphabet, and $\delta=\delta_c
\uplus \delta_r$ where $\delta_c \subseteq Q \times\Sigma_c \times
\Gamma\times Q$ are the {\em call transitions}, $\delta_r \subseteq Q
\times\Sigma_r \times\Gamma \times Q$ are the {\em return transitions}.
On a call transition $(q,a,q',\gamma)\in\delta_c$, $\gamma$ is pushed
onto the stack and the control goes from $q$ to $q'$.
On a return transition $(q, \gamma,a,q')\in\delta_r$, $\gamma$ is
popped from the stack.
Stacks are elements of $\Gamma^*$, and we denote by $\bot$ the
empty stack. A {\em run} of a \vpa $A$ on a word
$w=a_1 \ldots a_l$ is a sequence  $\{(q_k, \sigma_k)\}_{0\leq k\leq
  l}$, where $q_k$ is the state and $\sigma_k\in \Gamma^*$ is the
stack at step $k$, such that $q_0\in I$, $\sigma_0=\bot$,
and for each $k< l$, we have either: 
$(i)$ $(q_k,a_{k+1}, \gamma,q_{k+1})\in \delta_c$ and 
$\sigma_{k+1}=\sigma_k\gamma$; 
$(ii)$ $(q_k, a_{k+1},\gamma, q_{k+1})\in \delta_r$, and 
$\sigma_{k}=\sigma_{k+1}\gamma$. A run is {\em accepting} if $q_l\in
F$ and $\sigma_l=\bot$. A word $w$ is {\em accepted} by $A$ if there exists an accepting run
of $A$ over $w$. Note that it is necessarily well nested. $L(A)$, the {\em language} of $A$, is the set of
words accepted by $A$.  A language $L$ over $\Sigma$ is a {\em visibly
  pushdown language} if there is a \vpa $A$ over $\Sigma$ such that
$L(A)=L$.

In contrast to \cite{RM04} and to ease the notations, we do not allow transitions on the empty stack.
Therefore the words accepted by a \vpa are well-nested (every call symbol has a matching return symbol and conversely).

\paragraph{Visibly Pushdown Transducers} As finite-state transducers
extend finite-state automata with outputs, visibly pushdown
transducers extend \vpa with outputs. To simplify notations, 
we suppose that the output alphabet is $\Sigma$, but our results
still hold for an arbitrary output alphabet.

\begin{definition}[Visibly pushdown transducers]
\rm
A {\em visibly pushdown transducer}\footnote{In contrast to \cite{RS08}, there is no producing $\epsilon$-transitions
  (inserting transitions) but a transition may produce a word and not a single symbol}
 (\vpt) on finite words over
$\Sigma$ is a tuple $T=(Q,I, F, \Gamma, \delta)$ where 
$Q$ is a finite set of states, $I\subseteq Q$ is the set of initial
states, $F\subseteq Q$ the set of final states, $\Gamma$ is the stack alphabet, $\delta=\delta_c \uplus
\delta_r$ the transition relation, with
$\delta_c \subseteq Q \times\Sigma_c \times \Sigma^* \times
\Gamma \times Q$,
$\delta_r \subseteq Q \times\Sigma_r \times \Sigma^* \times
\Gamma\times Q$.
\end{definition}

A \textit{configuration} of a \vpt is a pair $(q,\sigma)\in Q\times \Gamma^*$. 
A \textit{run} of $T$ on a word $u = a_1\dots a_{l}\in\Sigma^*$ from a configuration 
$(q,\sigma)$ to a configuration $(q',\sigma')$ is a finite sequence
$\rho = \{(q_i,\sigma_i)\}_{0\leq k \leq l}$ such
that $q_0=q$, $\sigma=\sigma_0$, $q'=q_n$, $\sigma'=\sigma_n$ and
for all $i\in\{1,\dots,l\}$, there exist $v_i\in\Sigma^*$ and
$\gamma_i\in\Gamma$ such that $(q_i-1,a_i,v_i,\gamma_i,q_{i})\in
\delta_c$ and either $a_i\in\Sigma_c$ and $\sigma_{i}
= \sigma_{i-1}\gamma_i$, or $a_i\in \Sigma_r$ and $\sigma_{i-1}
= \sigma_{i}\gamma_i$. The word $v = v_1\dots v_{l}$ is 
called an \textit{output} of $\rho$. We write
$(q,\sigma)\xrightarrow{u/v} (q',\sigma')$ when there exists
a run on $u$ from $(q,\sigma)$ to $(q',\sigma')$ producing $v$ as output. 
The transducer $T$ defines a word binary relation 
$\inter{T} = \{ (u,v)\ |\ \exists q\in I,p\in
F,\ (q,\bot)\xrightarrow{u/v}(p,\bot)\}$.

The \textit{domain} of $T$, resp. the \textit{codomain} of $T$,
denoted resp. by $\dom(T)$ and $\img(T)$, is the domain of
$\inter{T}$, resp. the codomain of $\inter{T}$. Note that the domain
of $T$ contains only well nested words, which is not the case of the codomain in general.

In this paper, we prove the following theorem:

\begin{theorem}
    Functionality of \vpts is decidable in \pspace.
\end{theorem}

The rest of the paper is devoted to the proof of this theorem.

%% file: pumping.tex
\section{Preliminaries on Word Combinatorics}\label{sec:wordeq}

The size of a word $x$ is denoted by $|x|$.
Given two words $x,y\in \Sigma^*$, we write $x\preceq y$ if $x$ is a
prefix of $y$. If we have $x\preceq y$, then we note $x^{-1} y$ the
unique word $z$ such that $y=xz$. A word $x\in\Sigma^*$ is
\textit{primitive} if there is no word $y$ such that $|y|<|x|$ and 
$x\in y^*$. The \textit{primitive root} of a word $x\in\Sigma^*$
is the (unique) primitive word $y$ such that $x\in y^*$. In
particular, if $x$ is primitive, then its primitive root is $x$.
Two words $x$ and $y$ are \textit{conjugate} if there exists
$z\in\Sigma^*$ such that $xz = zy$. It is well-known that
two words are conjugate iff there exist $t_1,t_2\in\Sigma^*$ such
that $x = t_1t_2$ and $y = t_2t_1$. Two words $x,y\in\Sigma^*$
\textit{commute} iff $xy = yx$.

\begin{lemma}[folklore]\label{lemma:commute}
  Let $x,y\in\Sigma^*$ and $n,m\in\mathbb{N}$.
  \begin{enumerate}
  \item if $x$ and $y$ commute, then 
    $x,y\in z^*$ for some $z\in \Sigma^*$. Moreover, if $xy$ is
    primitive, then $x=\epsilon$ or $y=\epsilon$;
  \item if $x^n$ and $y^m$ have a common subword of length at least $|x| + |y| - d$ ($d$ being the
    greatest common divisor of $|x|$ and $|y|$), then their primitive
    roots are conjugate.
  \end{enumerate}
\end{lemma}

\begin{proof}
  The first assertion is folklore. 
  For the second, there exists $z\in\Sigma^*$ and $\alpha,\beta\geq 0$ such that
  $x = z^\alpha$ and $y=z^\beta$. If $x$ and $y$ are non-empty, then
  $\alpha,\beta > 0$ and $z\neq\epsilon$. Thus $xy =
  z^{\alpha+\beta}$, which contradicts the primitivity of $xy$. 
\end{proof}

\begin{lemma}[Hakala, Kortelainen, Theorem 7 of \cite{Hakala99}]\label{lem:finland}
    Let $v_0,v_1,v_m,v_{\ol 1},$ $v_{\ol 0},w_0,w_1,w_m,w_{\ol 1},w_{\ol
      m}\in \Sigma^*$ and $i\in\mathbb{N}$. If $v_0(v_1)^iv_m(v_{\ol
      1})^iv_{\ol 0} = w_0(w_1)^iw_m(w_{\ol 1})^iw_{\ol 0}$ holds for
    all $i\in\{0,1,2,3\}$, then it holds for all $i\in\mathbb{N}$.
\end{lemma}



\noindent Let $x\in\Sigma^*$, we denote by $x^\omega\in\Sigma^\omega$ the
infinite (countable) concatenation of $x$.

\begin{lemma}\label{lemma:overlap}
Let $x,x_1,x_2,y,z,t_1, t_2,p,q \in\Sigma^*$ with $t_1t_2, p,q$ primitive, then:
\begin{enumerate}
\item if $t_1\prec p$ and $xpt_1 = ypp$ then $xp^{\omega} = yp^\omega$ \label{lemma:overlap2:E}
\item if $xp^{\omega} = y p^\omega$ then $\exists \alpha, \beta\geq 0: xp^{\alpha} = y p^\beta$ \label{lemma:overlap2:O}
\item if $x(t_1t_2)^{\omega} = y (t_2t_1)^\omega$ and $t_1\neq \epsilon$, then $\exists \alpha, \beta\geq 0: x(t_1t_2)^{\alpha} = y (t_2t_1)^\beta t_2$ \label{lemma:overlap2:A}
\item if $x(t_1t_2)^{\omega} = (t_2t_1)^\omega$ and $t_1\neq\epsilon$,
  then $\exists \alpha \geq 0: x= (t_2t_1)^\alpha t_2$. \label{lemma:overlap2:B}
\item if $\forall i\in\{1,2\}$, $x_iy(t_1t_2)^{\omega} = y(t_1t_2)^\omega$ then 
$\exists \alpha_1,\alpha_2 \geq 0, \exists t_3,t_4\in\Sigma^*: t_3t_4=t_1t_2, x_i= (t_4t_3)^{\alpha_i}$ \label{lemma:overlap2:C}
\item if $xp^{\omega} = p^\omega$ then $\exists \alpha \geq 0: x= p^\alpha$ \label{lemma:overlap2:D}
\item if $\exists \alpha > 0$ such that $p^\alpha x
  p^\omega = x p^\omega$, then $x\in p^*$. \label{lemma:overlap2:F}
\item if $\exists \alpha>0,\ q^{\alpha} y p^{\omega} = y p^\omega$ then
  $qy=yp$ \label{lemma:overlap2:G}
\item if $\exists \alpha,\beta,\gamma\geq 1$ such that $x(t_1t_2)^\alpha y (t_1t_2)^\beta z
  = (t_2t_1)^\gamma$, then $y\in (t_1t_2)^*$\label{lem:cheese}.
\end{enumerate}
\end{lemma}

\begin{proof}

\begin{enumerate}
\item Let $t_2$ such that $p = t_1t_2$, then $xt_1t_2t_1 = y t_1t_2 t_1t_2 $, 
by Lemma \ref{lemma:commute} $t_1=\epsilon$ or $t_2=\epsilon$ i.e. 
either $t_1=\epsilon$ or $t_1=p$.

\item Direct consequence of the previous property since we have $xp^{\alpha}t_1 = y p^\beta$ for some $\alpha, \beta>1$ and $t_1\prec p$.

\item By applying the previous property to $x(t_1t_2)^{\omega} = y t_2 (t_1t_2)^\omega$.

\item The second assertion is a direct consequence of the first when taking
$y=\epsilon$. 

\item It is clear if $x_1 = x_2 = \epsilon$. Suppose that $x_1\neq
  \epsilon$. Since $x_1y(t_1t_2)^\omega = y(t_1t_2)^\omega$, we also
  have $x_1x_1y(t_1t_2)^\omega = y(t_1t_2)^\omega$, and more
  generally, for all $\beta\geq 1$, $(x_1)^\beta y (t_1t_2)^\omega =
  y(t_1t_2)^\omega$. By taking $\beta$ large enough, there exists
  $\gamma\geq 0$ such that $(x_1)^\beta$ and
  $(t_1t_2)^{\gamma}$ have a common factor of length at most
  $|x_1| + |t_1t_2| - \text{gcd}(|x_1|,|t_1t_2|)$. By the fundamental
  lemma, there exists $t_3,t_4\in\Sigma^*$ such that $t_3t_4$ is
  primitive, $x_1\in (t_4t_3)^*$ and $t_1t_2\in (t_3t_4)^*$. Since $t_1t_2$ is
  primitive, we have $t_1t_2 = t_3t_4$. Suppose that
  $x_2\neq\epsilon$. Similarly, we can prove that 
  $x_2 = (t'_4t'_3)^{\gamma}$ for some $\gamma > 0$ and
  $t'_3,t'_4$ such that $t_1t_2 = t'_3t'_4$. 
  We have $x_1y (t_1t_2)^\omega = x_2y(t_1t_2)^\omega$, 
  therefore $t_4t_3 = t'_4t'_3$, and $x_2\in (t_4t_3)^*$. 

\item We have $xp^\omega=p^\omega$ so we also have 
$px p^\omega=p^\omega$, therefore $xp^\omega=p xp^\omega$ i.e. 
 $xp = p x$, and by Lemma \ref{lemma:commute}, $x\in p^*$.

\item We clearly have $xp^\alpha=p^\alpha x$ therefore, by Lemma \ref{lemma:commute},
 $x\in p^*$.

\item We have $q^{\alpha} y p^{\omega} = y p^\omega$, this implies that for any $x\geq 0$ 
$q^{x \alpha} y p^{\omega} = y p^\omega$. 
Therefore, there exist $\beta\geq 0$ and $t_1 \prec q$ with $y=q^\beta t_1$. 
Let $t_2\in\Sigma^*$ such that $q=t_1t_2$, we have $(t_1t_2)^{\alpha+\beta}t_1=(t_1t_2)^{\beta}t_1 p^\alpha$.
Therefore because $|p|=|q|=|t_1t_2|$ we have $p=t_2t_1$. This concludes the proof.

\item We assume $t_1, t_2\neq \epsilon$ (otherwise it is obvious).
By \ref{lemma:overlap2:E} and \ref{lemma:overlap2:B} we have that $x= (t_2t_1)^at_2$. 
By the same argument we have $z= t_1(t_2t_1)^b$
So we have: $t_2(t_1t_2)^{\alpha+a} y (t_1t_2)^{\beta+b} t_1=(t_2t_1)^\gamma$.
Therefore $y\in(t_1t_2)^*$.

\end{enumerate}
\qed
\end{proof}

\section{From Functionality to Word Equations}\label{sec:pumpingfun}

Given some words $u_0, \dots u_n, u_m, u_{\ol n}, \dots, u_{\ol
  0}\in\Sigma^*$, $k\in\mathbb{N}$, and a function
$\pi:\{1,\dots,k\}\rightarrow \{1,\dots,n\}$, we denote by $u_\pi$ the
word $u_0u_{\pi(1)}\dots u_{\pi(j)}u_m u_{\ol \pi(j)}\dots u_{\ol
  \pi(1)} u_{\ol 0}$.  We denote by $id_n$ the identity function on
domain $\{1,\dots,n\}$. The following lemma states that if a word $u$
translated into two words $v,w$ is high enough, $u$, $v$ and $w$ can
be decomposed into subwords that can be removed, repeated, or
permutted in parallel in $u$, $v$ and $w$, while preserving the
transduction relation.

\begin{lemma}\label{lem:verticalpumping}
  Let $T$ be a \vpt with $N$ states, and $n\geq 1$.  Let
  $u,v,w\in\Sigma^*$ such that $v,w\in T(u)$ ($u$ is thus well nested)
  and $h(u)>nN^4$.  Then there exist $u_m,v_m,w_m \in\Sigma^*$ and
  $u_i,u_{\ol i},v_i,v_{\ol i},w_i,w_{\ol i}\in\Sigma^*$ for all
  $i\in\{0,\dots,n\}$ such that $u_{id_n} = u$, $v_{id_n} = v$,
  $w_{id_n} = w$ and for all $k\in \mathbb{N}$ and all
  $\pi:\{1,\dots,k\}\rightarrow \{1,\dots,n\}$: $v_\pi,w_\pi\in
  T(u_\pi)$ and $u_i,u_{\ol i}\neq\epsilon$ for all $i=1,\dots,n$.
\end{lemma}

\begin{proof}
  Let $T$ be a \vpt, with set of states $Q$. Let $N=|Q|$, $n\geq 1$,
  and $u,v,w\in\Sigma^*$ such that $v,w\in T(u)$ and $h(u)>nN^4$. In
  particular, $u$ is well nested. We denote by $\ell$ the length of
  the word $u$ and write $u=(a_j)_{1\leq j \leq \ell}$, with $a_j\in
  \Sigma$ for all $j$. There exists a position $1\leq j \leq \ell$ in
  $u$ whose height is equal to $h(u)$. We fix such a position
  $j$. Then, for any height $0\leq k \leq h(u)$, we define two
  positions, denoted $\alpha(k)$ and $\beta(k)$. $\alpha(k)$
  (resp. $\beta(k)$) is the largest (resp. the smallest) index $d$,
  such that $d \leq j$ (resp. $d\geq j$) and the height of $u$ in
  position $d$ is equal to $k$. The part of the word concerned by
  mapping $\alpha$ (resp. $\beta$) is represented in blue (resp. in
  red) on Figure~\ref{fig:pumpingvertical}.
  
\input{figvertical}

  As $v,w\in T(u)$, there exists two runs $\varrho_v, \varrho_w$ on
  $u$ in $T$ which produce respectively the outputs $v$ and $w$. We
  denote by $(p_i)_{0\leq i \leq \ell}$ (resp. $(q_i)_{0\leq i \leq
    \ell}$) the states we encounter along $\varrho_v$
  (resp. $\varrho_w$). As $h(u)>nN^4$, there exists two pairs of states 
  $(p,p'), (q,q')\in Q^2$ such that 
$$
|\{0\leq k \leq h(u) \mid p_{\alpha(k)}=p \text{ and }
p_{\beta(k)}=p' \text{ and }
q_{\alpha(k)}=q \text{ and }
q_{\beta(k)}=q'\}| > n
$$ 
We denote by $0\leq k_1<\ldots<k_{n+1}\leq h(u)$ the $n+1$ different
heights associated with the pairs $(p,p')$ and $(q,q')$. For each
$i=0,\ldots,n-1$, this means that the two runs pass simultaneously in
states $p$ and $q$ before a call transition with a height equal to
$k_i$, and that the height of the stack will never be smaller than
$k_i$, until reaching again states $p$ and $q$ with a stack of height
$k_{i+1}$. A symmetric property can be stated for states $p'$ and
$q'$. As a consequence, we obtain $n$ fragments which behave as
synchronized ``call loops'' around $p$ and $q$ with corresponding
``return loops'' around $p'$ and $q'$. This situation is described on
Figure~\ref{fig:pumpingvertical}.

Then, we can define the different fragments of $u$ as follows: (see
Figure~\ref{fig:pumpingvertical})
\begin{itemize}
\item $u_0 = a_1 \ldots a_{\alpha(k_1)-1}$,
\item $\forall 1\leq i \leq n, u_i = a_{\alpha(k_i)} \ldots a_{\alpha(k_{i+1})-1}$,
\item $u_m = a_{\alpha(k_{n+1})}\ldots a_{\beta(k_{n+1})-1}$,
\item $\forall 1\leq i \leq n, u_{\ol i} = a_{\beta(k_{i+1})} \ldots a_{\beta(k_{i})-1}$,
\item $u_{\ol 0} = a_{\beta(k_1)} \ldots a_\ell$.
\end{itemize}

We immediately obtain $u=u_{id_n}$ and $u_i,u_{\ol i} \neq\epsilon$
for all $i=1,\dots,n$. The decompositions of $v$ and $w$ are obtained
by considering the outputs produced by the corresponding fragments of
$u$ on the two runs $\varrho_v$ and $\varrho_w$.

Finally, the property of commutativity ( $v_\pi,w_\pi\in T(u_\pi)$ for
all $\pi:\{1,\dots,k\}\rightarrow \{1,\dots,n\}$) easily follows from
the fact that for each $i\in \{1,\dots,n\}$, the fragments of the runs
associated with $u_i$ and $u_{\ol i}$ do not depend on the content of
the stack as $T$ is a \emph{visibly} pushdown transducer.
\qed
\end{proof}



The following lemma states that if a word $u$ with at least two
outputs is high enough, there is a word $u'$ strictly less higher with
at least two outputs.

\begin{lemma}\label{lem:pumpvert}
  Let $T$ be a \vpt with $N$ states and $u\in \dom(T)$ such that
  $|T(u)|>1$ and $h(u)> 8N^4$. There
  exists $u'\in \dom(T)$ such that $|T(u')|\geq 2$ and $|u'| < |u|$.
\end{lemma}

\begin{proof}
    Let $v,w\in T(u)$ such that $v\neq w$. Thanks to
    Lemma \ref{lem:verticalpumping}, there exist $u_m,v_m,w_m \in\Sigma^*$, and 
    for all $i\in\{0,\dots,8\}$, there exist $u_i,u_{\ol i},v_i,v_{\ol i},w_i,w_{\ol i}\in\Sigma^*$, 
    such that $u_{id_8} = u$, $v_{id_8} = v$, $w_{id_8} = w$ and for all
    $k\in \mathbb{N}$ and all $\pi:\{1,\dots,k\}\rightarrow \{1,\dots,n\}$:
    $v_\pi,w_\pi\in T(u_\pi)$ and $u_i,u_{\ol i}\neq\epsilon$ for all $i=1,\dots,n$.
    We prove that there exist $k\in\{0,\dots,7\}$ and $\pi:\{1,\dots,
    j\}\rightarrow \{1,\dots,{8}\}$ such that $v_\pi\neq w_\pi$ and $|u_\pi|<|u|$.
    We proceed by contradiction. Suppose that for all $k\in\{0,\dots,7\}$ and for all 
    $\pi:\{1,\dots,k\}\rightarrow \{1,\dots,{8}\}$ such that $|u_\pi|<|u|$ we
    have $v_\pi = w_\pi$. This defines a system of equations
    $\Sys = \{ v_\pi = w_\pi\ |\ \pi:\{1,\dots,k\}\rightarrow \{1,\dots,8\},\ |u_\pi|< |u|\}$.
    We show in the next section that it implies $v=w$ (Theorem \ref{thm:sys}).
\end{proof}

%% file: figvertical.tex
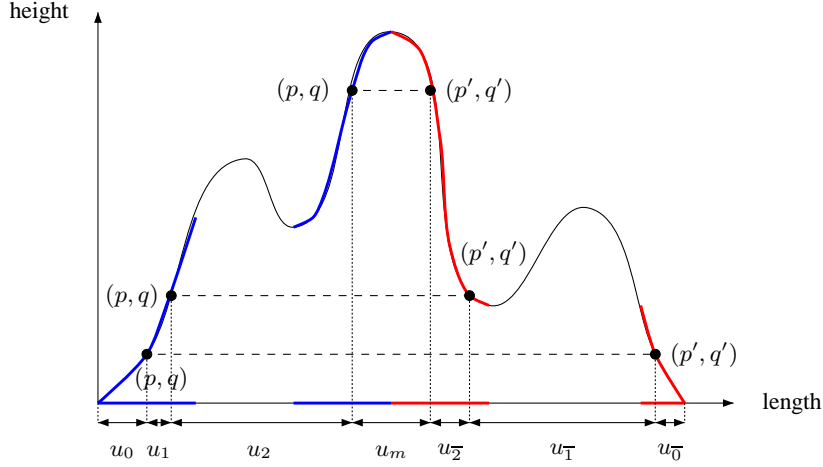
\begin{figure}[!ht]
\centering
\begin{picture}(90,62)(0,-10)
  
  \unitlength=1.3mm
  \drawcurve[AHnb=0](0,0)(5,5)(15,25)(20,18)(30,38)(40,10)(50,20)(57,5)(60,0)
  \drawcurve[AHnb=0,linewidth=0.3,linecolor=blue](0,0)(5,5)(7,10)(10,19)
  \drawcurve[AHnb=0,linewidth=0.3,linecolor=blue](20,18)(22,19)(23,21)(26,32)(28,37)(30,38)

  \drawcurve[AHnb=0,linewidth=0.3,linecolor=red](30,38)(32,37)(33,36)(35,27)(36,16)(38,11)(40,10)
  \drawcurve[AHnb=0,linewidth=0.3,linecolor=red](55.5,10)(57,5)(60,0)

  \drawline(0,0)(0,40)
  \drawline(0,0)(65,0)

  \drawline[AHnb=0,linewidth=0.3,linecolor=blue](0,0)(10,0)
  \drawline[AHnb=0,linewidth=0.3,linecolor=blue](20,0)(30,0)

  \drawline[AHnb=0,linewidth=0.3,linecolor=red](30,0)(40,0)
  \drawline[AHnb=0,linewidth=0.3,linecolor=red](55.5,0)(60,0)

  \node[Nframe=n,NLdist=6,NLangle=180](t1)(0,40){height}
  \node[Nframe=n,NLdist=6,NLangle=0](t2)(65,0){length}

  \node[Nw=1,Nh=1,fillcolor=black,Nmr=1,NLdist=3,NLangle=300](p1)(5,5){$(p,q)$}
  \node[Nw=1,Nh=1,fillcolor=black,Nmr=1,NLdist=4,NLangle=180](p2)(7.5,11){$(p,q)$}
  \node[Nw=1,Nh=1,fillcolor=black,Nmr=1,NLdist=5,NLangle=180](p3)(26,32){$(p,q)$}
  
  \node[Nw=1,Nh=1,fillcolor=black,Nmr=1,NLdist=5,NLangle=0](q1)(57,5){$(p',q')$}
  \node[Nw=1,Nh=1,fillcolor=black,Nmr=1,NLdist=5,NLangle=60](q2)(38,11){$(p',q')$}
  \node[Nw=1,Nh=1,fillcolor=black,Nmr=1,NLdist=5,NLangle=0](q3)(34,32){$(p',q')$}

  \drawedge[dash={0.8}1,AHnb=0](p1,q1){}
  \drawedge[dash={0.8}1,AHnb=0](p2,q2){}
  \drawedge[dash={0.8}1,AHnb=0](p3,q3){}

  \gasset{Nframe=n,Nmr=0,Nw=0,Nh=0}

  \node[Nframe=n](b1)(5,-2){}
  \node[Nframe=n](b2)(7.5,-2){}
  \node[Nframe=n](b3)(26,-2){}
  \node[Nframe=n](b4)(34,-2){}
  \node[Nframe=n](b5)(38,-2){}
  \node[Nframe=n](b6)(57,-2){}

  \drawedge[dash={0.2}1,AHnb=0](b1,p1){}
  \drawedge[dash={0.2}1,AHnb=0](b2,p2){}
  \drawedge[dash={0.2}1,AHnb=0](b3,p3){}  
  \drawedge[dash={0.2}1,AHnb=0](b4,q3){} 
  \drawedge[dash={0.2}1,AHnb=0](b5,q2){} 
  \drawedge[dash={0.2}1,AHnb=0](b6,q1){}  
  \drawline[dash={0.2}1,AHnb=0](0,0)(0,-2)
  \drawline[dash={0.2}1,AHnb=0](60,0)(60,-2)

  \drawline[ATnb=1](0,-2)(5,-2)
  \node(u0)(2.5,-5){$u_0$}

  \drawline[ATnb=1](5,-2)(7.5,-2)
  \node(u1)(6.25,-5){$u_1$}

  \drawline[ATnb=1](7.5,-2)(26,-2)
  \node(u2)(16.5,-5){$u_2$}

  \drawline[ATnb=1](26,-2)(34,-2)
  \node(um)(30,-5){$u_m$}

  \drawline[ATnb=1](34,-2)(38,-2)
  \node(um)(36,-5){$u_{\ol 2}$}

  \drawline[ATnb=1](38,-2)(57,-2)
  \node(um)(47.5,-5){$u_{\ol 1}$}

  \drawline[ATnb=1](57,-2)(60,-2)
  \node(um)(58.5,-5){$u_{\ol 0}$}

\end{picture}
\caption{Form of pumping}
\label{fig:pumpingvertical}
\end{figure}

%% file: equations.tex
\section{Word Equations}\label{sec:equations}

In this section, we fix some $n\geq 8$, some
words $u_m,v_m,w_m\in\Sigma^*$ and for all $i\in\{0,\dots,n\}$, 
we fix $u_i,v_i,w_i,u_{\ol i},v_{\ol i},w_{\ol i}\in \Sigma^*$ such that
$u_i,u_{\ol i}\neq\epsilon$. We consider the system
$\Sys = \{ v_\pi = w_\pi\ |\ \pi:\{1,\dots,k\}\rightarrow \{1,\dots,n\},\ |u_\pi|< |u_{id_n}|\}$.
The main result we prove is the following:

\begin{theorem}\label{thm:sys}
    If $\Sys$ holds, then $v_{id_n} = w_{id_n}$.
\end{theorem}

We let $\ell\in\{1,\dots,n\}$ such that $|u_\ell u_{\ol \ell}|\leq |u_iu_{\ol i}|$ for
all $i\in\{1,\dots,n\}$. We consider several cases to prove Theorem \ref{thm:sys}:
$$
(1)\ |v_\ell| = |w_\ell|\qquad (2)\ |v_\ell| > |w_\ell|\qquad (3)\ |w_\ell| > |v_\ell|
$$

\noindent Cases $2$ and $3$ being symmetric, we consider cases $1$ and $2$ only in
the two following subsections.

\subsection{Proof of Theorem \ref{thm:sys}: case $|v_\ell|>|w_\ell|$}

We denote by $\Sys[|v_\ell|>|w_\ell|]$ the system $\Sys$ with the assumption $|v_\ell|>|w_\ell|$ and from
now one we assume that this system holds. We consider the following set of
equations, defined for all $a,b\geq 0$ and all $i\in\{1,\dots,n\}$:

$$
\left\{\begin{array}{rclr}
v_0v_mv_{\ol 0} & = & w_0w_mw_{\ol 0} & (1)\\
v_0(v_\ell)^av_m(v_{\ol \ell})^av_{\ol 0} & = & w_0(w_\ell)^aw_m(w_{\ol \ell})^aw_{\ol 0} & (2)\\
v_0v_i(v_\ell)^av_m(v_{\ol \ell})^av_{\ol i}v_{\ol 0} & = &
w_0w_i(w_\ell)^aw_m(w_{\ol \ell})^aw_{\ol i}w_{\ol 0}  & (3) \\
v_0(v_\ell)^av_i(v_\ell)^bv_m(v_{\ol \ell})^bv_{\ol i}(v_{\ol \ell})^av_{\ol 0} & = &
w_0(w_\ell)^aw_i(w_\ell)^bw_m(w_{\ol \ell})^bw_{\ol i}(w_{\ol \ell})^aw_{\ol 0} & (4) \\
\end{array}\right.
$$

For $k\in\{1,2,3,4\}$, we denote by $\Sys_k$ the subsystem that
of equations of type $k$. For instance, $\Sys_2$ is the system
of equations $\{v_0(v_\ell)^av_m(v_{\ol \ell})^av_{\ol 0} = 
w_0(w_\ell)^aw_m(w_{\ol \ell})^aw_{\ol 0}\ |\ a\in\mathbb{N}\}$.

\begin{lemma}
For all $k\in\{1,\dots,4\}$, $\Sys_k$ holds.
\end{lemma}

\begin{proof}
First, $|u_0u_mu_{\ol 0}| < |u_{id_n}|$ and $u_0u_mu_{\ol 0} = u_\pi$ where
$\pi$ is the function with empty domain. Since $\Sys$ holds by hypothesis, 
this equation holds.

We prove that $\Sys_4$ holds, as $\Sys_3$ is a particular case of $\Sys_4$ and
$\Sys_2$ is a similar but easier case. First, $\Sys_4$ holds for all $a,b\in\{0,1,2,3\}$. Indeed, since $n\geq 8$, there are six pairwise
different integers $i_1,\dots,i_6\in\{1,\dots,n\}$ such that
$i_k\neq i$ for all $k\in\{1,\dots,6\}$ and
$6|u_{\ell} u_{\ol \ell}| + |u_iu_{\ol i}| \leq |u_iu_{\ol i}| + \sum_{k=1}^6
|u_{i_k}u_{\ol {i_k}}| < |u_{id_n}|$.
Second, by Lemma \ref{lem:finland}, $\Sys_4$ holds for all $a\in\mathbb{N}$ and $b=0,1,2,3$.
If we fix $a_0\in\mathbb{N}$, it holds for $a=a_0$ and $b=0,1,2,3$. Thus by Lemma \ref{lem:finland}
it holds for $a=a_0$ and all $b\in\mathbb{N}$.\qed
\end{proof}


\begin{proposition}\label{prop:lengtheq}
    For all $i\in\{1,\dots,n\}$, $|v_iv_{\ol i}| = |w_iw_{\ol i}|$.
\end{proposition}

\begin{proof}
    This is implied by $\Sys_1$ and $\Sys_4$ (with $a=b=0$).\qed
\end{proof}

Thanks to $\Sys_1,\dots,\Sys_4$ we can characterize the form of $v_i,w_i,w_{\ol i}$
for all $i$ and prove a property on $v_m,w_m$. This characterization is then used
to prove $v_{id_n} = w_{id_n}$. Wlog we assume that 
$v_0=\epsilon$ or $w_0=\epsilon$, and $v_{\ol 0}=\epsilon$ or $w_{\ol 0} = \epsilon$.
Otherwise we can remove their common prefixes in $\Sys_1,\dots,\Sys_4$. 

\begin{lemma}\label{lem:allform}
    If there exist $k\in\{1,\dots,n\}$ such that $w_k\neq\epsilon$.
    Then there exist
    $t_1,t_2,t_3,t_4\in\Sigma^*$, $\alpha_0,\beta_0\geq 0$, 
    $\alpha_i,\beta_i,\beta_{\ol i}\geq 0$ for all $i\in\{1,\dots,n\}$
    such that $t_1t_2$ is primitive and for all $i\in\{1,\dots,n\}$:
    $$
    \begin{array}{c@{\quad}c@{\quad}c@{\quad}c@{\quad}ccccc}
    t_1t_2 = t_3t_4 & t_4t_3w_m = w_mt_2t_1 & v_i  =
    (t_1t_2)^{\alpha_i} &
    w_i  =  (t_4t_3)^{\beta_i} &
    w_{\ol i}  =  (t_2t_1)^{\beta_{\ol i}}
    \end{array}
    $$
    and if $w_0 = \epsilon$, then $v_0  = (t_4t_3)^{\alpha_0}t_4$, and
    if $v_0 = \epsilon$, then  $w_0 =(t_3t_4)^{\beta_0}t_3$.
\end{lemma}

\begin{proof}
  First we infer the form of $v_{\ell}$ and $w_{\ol \ell}$.
  Since $|v_\ell| > |w_{\ell}|$, by $\Sys_2$, there is $a\geq 0$ such that
  $(v_{\ell})^a$ and $(w_{\ol \ell})^a$ have a common factor of length
  at least $|v_{\ell}| + |w_{\ol \ell}| - \text{gcd}(|v_{\ell}|,
  |w_{\ol \ell}|)$ (see Fig.~\ref{figS2}). Therefore by Lemma \ref{lemma:commute}.2,
  there exist $t_1,t_2\in \Sigma^*$ such that
  $t_1t_2$ is primitive, $v_\ell\ =\ (t_1t_2)^{\alpha_\ell}$ and
  $w_{\ol \ell}\ =\ (t_2t_1)^{\beta_{\ol \ell}}$ for some
  $\alpha_{\ell},\beta_{\ol \ell} > 0$. 
  \input{figS2}


  Second we derive the form of $v_i$ and $w_{\ol i}$ for all
  $i\in\{1,\dots,n\}$. As $v_\ell\neq \epsilon$ there is $b_0\geq 1$ such that 
  $|(v_\ell)^{b_0-1} v_m (v_{\ol \ell})^{b_0} v_{\ol i}v_{\ol 0}| \geq   |w_{\ol
    0}|$. We consider $\Sys_4$ with $b=b_0$. 
    The size of the suffix $v_\ell v_i(v_\ell)^{b_0} v_m (v_{\ol \ell})^{b_0} v_{\ol i}(v_{\ol \ell})^{a}v_{\ol 0}$ 
    is of the form $l_1(a)=k_1+a|v_{\ol \ell}|$ and the size of the suffix $(w_{\ol \ell})^a w_{\ol 0}$ is of the form
    $l_2(a) = k_2+a|w_{\ol \ell}|$. As $|w_{\ol \ell}|> |v_{\ol \ell}|$ (by Proposition \ref{prop:lengtheq} and $|v_{ \ell}|> |w_{\ell}|$),
    there exists $a_0\geq 1$ such that $l_2(a_0)\geq l_1(a_0)$. 
    Therefore (see Fig.~\ref{figS5}) $v_\ell v_i v_\ell$ is a factor of
  $(w_{\ol \ell})^{a_0}$. Thus there is $X,Z\in\Sigma^*$ such that
  $X(t_1t_2)^{\alpha_{\ell}} v_i (t_1t_2)^{\alpha_{\ell}} Z =
  (t_2t_1)^{a_0\beta_{\ol \ell}}$. Since $\alpha_{\ell},\beta_{\ol \ell}>0$, we can
  apply Lemma \ref{lemma:overlap}.\ref{lem:cheese} and we get $v_i\in (t_1t_2)^*$.
  Since $|w_{\ol \ell}|>|v_{\ol \ell}|$ and $w_{\ol \ell} = (t_2t_1)^{\beta_{\ol \ell}}$, by symmetry, 
  we also get $w_{\ol i}\in (t_2t_1)^*$.

  \input{figS5}

  Third we determine the form of the words $w_i$ and prove the
  property on $w_m$. Since $v_\ell =
  (t_1t_2)^{\alpha_\ell}$, $w_{\ol \ell} = (t_2t_1)^{\beta_{\ol
      \ell}}$ and $v_i = (t_1t_2)^{\alpha_i}$ for some $\alpha_i\geq
  0$, $\Sys_2$ and
  $\Sys_3$ can be rewritten as follows:
  \begin{eqnarray*}
v_0(t_1t_2)^{a.\alpha_{\ell}}v_m(v_{\ol \ell})^av_{\ol 0} & = & w_0(w_\ell)^aw_m(t_2t_1)^{a.\beta_{\ol \ell}}w_{\ol 0}\\
v_0(t_1t_2)^{\alpha_i + a.\alpha_{\ell}}v_m(v_{\ol \ell})^av_{\ol
  i}v_{\ol 0} & = & w_0w_i(w_\ell)^aw_m(t_2t_1)^{a.\beta_{\ol \ell}}w_{\ol i}w_{\ol 0}
  \end{eqnarray*}
  
  \noindent Since $|v_\ell| > |w_\ell|$, there exist $\alpha,\beta,\gamma,\gamma'\geq 2$ and $t',t''\prec t_1t_2$ such that
  $$
  v_0(t_1t_2)^{\alpha}t' =  w_0(w_\ell)^\beta w_mt_2(t_1t_2)^{\gamma}
  \qquad v_0(t_1t_2)^{\alpha}t'' = 
    w_0w_i(w_\ell)^\beta w_mt_2 (t_1t_2)^{\gamma'}
    $$

  \noindent   By Lemma
  \ref{lemma:overlap}.\ref{lemma:overlap2:E}, we get
  $v_0(t_1t_2)^{\omega} = w_0w_i(w_\ell)^\beta w_m(t_2t_1)^{\omega}$ and 
  \begin{eqnarray}
  v_0(t_1t_2)^{\omega}  =  w_0(w_\ell)^\beta w_m(t_2t_1)^{\omega}\  \label{eqn:aaa}
  \end{eqnarray}
  \noindent Therefore \begin{eqnarray}
      (w_\ell)^\beta w_m(t_2t_1)^{\omega} = w_i(w_\ell)^\beta w_m(t_2t_1)^{\omega}\ \label{eqn:eqq}
  \end{eqnarray}

  Eq. \ref{eqn:eqq} is equivalent to $(w_\ell)^\beta w_mt_2(t_1t_2)^{\omega} =
  w_i(w_\ell)^\beta w_mt_2(t_1t_2)^{\omega}$, thus by
  Lemma \ref{lemma:overlap}.\ref{lemma:overlap2:C}, there exist $t_3,t_4\in\Sigma^*$ such
  that $t_1t_2 = t_3t_4$ and for all $i\in\{1,\dots,n\}$, $w_i =
  (t_4t_3)^{\beta_i}$ for some $\beta_i\geq 0$. By hypothesis, there
  is $k\in\{1,\dots,n\}$ such that $w_k\neq \epsilon$, and therefore
  $\beta_k > 0$. Eq. \ref{eqn:eqq} gives  $(t_4t_3)^{\beta_\ell} w_m(t_2t_1)^{\omega}  =  (t_4t_3)^{\beta_\ell
    + \beta_k} w_m(t_2t_1)^{\omega}$, i.e. $w_m(t_2t_1)^{\omega}  =
  (t_4t_3)^{\beta_k} w_m(t_2t_1)^{\omega}$. By lemma \ref{lemma:overlap}.\ref{lemma:overlap2:G} we get $w_mt_2t_1 = t_4t_3w_m$.

  Finally, we determine the form of $v_0$ and $w_0$.
  If $w_0 = \epsilon$, then
  Eq. \ref{eqn:aaa} gives $v_0(t_1t_2)^{\omega} =
  (t_4t_3)^{\beta.\beta_{\ell}}w_m(t_2t_1)^\omega$.
  Since $t_1t_2 = t_3t_4$ and $t_4t_3w_m = w_mt_2t_1$, 
  $v_0(t_3t_4)^\omega = (t_4t_3)^\omega$. Wlog we can assume that
  $t_3\neq\epsilon$. Indeed, $v_\ell \in (t_1t_2)^*$ is non-empty
  and $t_1t_2 = t_3t_4$, so that $t_3t_4\neq\epsilon$. By Lemma
  \ref{lemma:overlap}.\ref{lemma:overlap2:B}, 
  $v_0\in (t_4t_3)^*t_4$. Alike, if $v_0 = \epsilon$, then
  wlog we can suppose that $t_4\neq \epsilon$, and
  conclude similarly that $w_0\in (t_3t_4)^*t_3$.\qed
\end{proof}

The mirror of a word $t\in\Sigma^*$ is denoted by $\overline{t}$ and is
inductively defined by $\ol{\epsilon} = \epsilon$, $\ol{ct} = \ol{t}c$ for all
$c\in\Sigma$. The mirror of an equation $t = t'$ is $\ol{t} = \ol{t'}$. By taking the
mirror of the equations $\Sys_1,\dots,\Sys_4$, we obtain a system of equations
which has the same form as $\Sys_1,\dots,\Sys_4$. Since $|v_{\ell}| > |w_{\ell}|$, by
Prop. \ref{prop:lengtheq}, $|w_{\ol \ell}|>|v_{\ol \ell}|$. Therefore we can apply
Lemma \ref{lem:allform} on the mirrors of $\Sys_1,\dots,\Sys_4$ and obtain
the following corollary:

\begin{corollary}\label{lem:allformmirror}
    If there exist $k\in\{1,\dots,n\}$ such that $v_{\ol k}\neq\epsilon$.
    Then there exist $t_1,t_2,t_5,t_6\in\Sigma^*$, $\alpha_0,\beta_0\geq 0$, 
    $\alpha_i,\beta_i,\beta_{\ol i}\geq 0$ for all $i\in\{1,\dots,n\}$
    such that $t_2t_1$ is primitive and for all $i\in\{1,\dots,n\}$:
    $$
    \begin{array}{c@{\quad}c@{\quad}c@{\quad}c@{\quad}ccccc}
    t_2t_1 = t_6t_5 & t_1t_2v_m = v_mt_5t_6 & v_i  =
    (t_1t_2)^{\alpha_i} &
    v_{\ol i}  =  (t_5t_6)^{\alpha_{\ol i}} &
    w_{\ol i}  =  (t_2t_1)^{\beta_{\ol i}}
    \end{array}
    $$
    and if $w_{\ol 0} = \epsilon$, then $v_{\ol 0}  = t_5(t_6t_5)^{\alpha_{\ol 0}}$, and if $v_{\ol 0} = \epsilon$, then  $w_{\ol 0} =t_6(t_5t_6)^{\beta_{\ol 0}}$
\end{corollary}

We are now equipped to prove that $v_{id_n} = w_{id_n}$:

\begin{theorem}\label{thm:syseq}
  $\Sys[|v_\ell|>|w_\ell|] \implies v_0\dots v_nv_mv_{\ol n}\dots v_{\ol 0} = w_0\dots w_nw_mw_{\ol n}\dots w_{\ol 0}$
\end{theorem}

\begin{proof}
    We consider several cases:

\begin{enumerate}

\item \textit{there exist $k,k'\in\{1,\dots,n\}$ such that
$w_{k'}\neq \epsilon$ and $v_{\ol k}\neq \epsilon$}.

By Lemma \ref{lem:allform}, there exist $t_1,t_2,t_3,t_4\in\Sigma^*$
and $\alpha_0,\beta_0,\dots,\alpha_n,\beta_n,\beta_{\ol
  n},\dots,\beta_{\ol 1}\geq 0$ such that:
$$
\begin{array}{ccccccccc}
  t_1t_2 = t_3t_4 \quad t_4t_3w_m = w_mt_2t_1 \quad v_i  = (t_1t_2)^{\alpha_i} \quad  w_i  =  (t_4t_3)^{\beta_i} \quad 
  w_{\ol i}  =  (t_2t_1)^{\beta_{\ol i}}
\end{array}
$$
and if $w_0 = \epsilon$, then $v_0  = (t_4t_3)^{\alpha_0}t_4$, and if $v_0 = \epsilon$, then  $w_0 =(t_3t_4)^{\beta_0}t_3$

By Corollary \ref{lem:allformmirror} and the fact that a word is uniquely 
decomposed as a power of a primitive word, there exist
$t_5,t_6\in\Sigma^*$ and $\alpha_{\ol n},\dots,\alpha_{\ol 1}\geq 0$
such that: 
$$
t_2t_1 = t_6t_5 \quad \quad t_1t_2v_m = v_mt_5t_6 \quad v_{\ol i}  = (t_5t_6)^{\alpha_{\ol i}}
$$
and if $w_{\ol 0} = \epsilon$, then $v_{\ol 0}  = t_5(t_6t_5)^{\alpha_{\ol 0}}$, and if $v_{\ol 0} = \epsilon$, then  $w_{\ol 0} =t_6(t_5t_6)^{\beta_{\ol 0}}$





    \noindent We can also suppose that $v_0 = (t_3t_4)^{\alpha_0} = (t_1t_2)^{\alpha_0}$ and $w_0
    =(t_3t_4)^{\beta_0}t_3$. Indeed, if $w_0 = \epsilon$,
    we simply replaced $v_0$ by  $t_3v_0$ and $w_0$ by $t_3w_0$. 
    Similarly, we assume that $w_{\ol 0} = (t_6t_5)^{\beta_{\ol
        0}}$ and $v_{\ol 0} = t_5(t_6t_5)^{\alpha_{\ol 0}}$.
    By Prop \ref{prop:lengtheq}, $\alpha_i+\alpha_{\ol i} = \beta_i+\beta_{\ol i}$ for all
    $i\in\{1,\dots,n\}$. Finally:
    $$
    \begin{array}{cllll}
      & & v_0v_1\dots v_nv_mv_{\ol n}\dots v_{\ol 0}  \\
      & = &
      (t_1t_2)^{\alpha_0+\dots+\alpha_n}v_m (t_5t_6)^{\alpha_{\ol n}+\dots+\alpha_{\ol 0}}t_5 \\
      & = &
      (t_1t_2)^{\alpha_0+\beta_1+\dots+\beta_n}v_m (t_5t_6)^{\beta_{\ol
          n}+\dots+\beta_{\ol 1}+\alpha_{\ol 0}}t_5 &
      \text{(since } \alpha_i+\alpha_{\ol i}=\beta_i+\beta_{\ol
        i} \text{ and } \\
      & & & t_1t_2v_m = v_mt_5t_6) \\
      & = &
      (t_1t_2)^{\beta_1+\dots+\beta_n}v_0 v_m  v_{\ol 0} (t_6t_5)^{\beta_{\ol n}+\dots+\beta_{\ol 1}}\\
      & = &
      (t_1t_2)^{\beta_1+\dots+\beta_n}w_0 w_m w_{\ol 0} (t_6t_5)^{\beta_{\ol n}+\dots+\beta_{\ol 1}}
       & \text{(by } \Sys_1 \text{)}\\
       & = &
      (t_1t_2)^{\beta_1+\dots+\beta_n}(t_3t_4)^{\beta_0}t_3 w_m
      (t_6t_5)^{\beta_{\ol 0}}(t_6t_5)^{\beta_{\ol n}+\dots+\beta_{\ol 1}}\\
      & = &
      (t_3t_4)^{\beta_0+\beta_1+\dots+\beta_n} t_3 w_m
      (t_2t_1)^{\beta_{\ol n}+\dots+\beta_{\ol 1} + \beta_{\ol 0}} & \text{(as
        $t_1t_2 = t_3t_4$ and $t_2t_1 = t_6t_5$)}\\
      & = &
      w_0w_1\dots w_nw_mw_{\ol n}\dots w_{\ol 1}w_{\ol 0} & \hfill \square\\
  \end{array}
    $$

    \item for all $k\in\{1,\dots,n\}$, $w_k = v_{\ol k} = \epsilon$.
      As in the proof of Lemma \ref{lem:allform}, we can characterize
      the form of $v_i$ and $w_{\ol i}$ for all $i\in\{1,\dots,n\}$. In particular,
      there exists $t_1,t_2\in\Sigma^*$ such that $t_1t_2$ is primitive and
      $v_i = (t_1t_2)^{\alpha_i}$ for some $\alpha_i\geq 0$, and
      $w_{\ol i} = (t_2t_1)^{\beta_{\ol i}}$ for some $\beta_i\geq
      0$. By Proposition \ref{prop:lengtheq}, $\alpha_i=\beta_i$ for
      all $i$. We let $w'_0 = w_0w_m$ and $v'_{\ol 0} = v_mv_{\ol
        0}$. The systems $\Sys_1,\Sys_2$ can therefore
      be rewritten as follows:
      $$
      \left\{\begin{array}{rclr}
      v_0v'_{\ol 0} & = & w'_0w_{\ol 0} & (1)\\
      v_0(t_1t_2)^{a\alpha_\ell}v'_{\ol 0} & = & w'_0(t_2t_1)^{a\alpha_\ell}w_{\ol 0} & (2)\\
      \end{array}\right.
      $$
      
      Wlog, we can assume that $v_0=\epsilon$ or $w'_0=\epsilon$. Both
      cases are symmetric, so that we consider only the case
      $v_0 = \epsilon$. Wlog we can assume that $t_1\neq\epsilon$.
      By Lemma \ref{lemma:overlap}.\ref{lemma:overlap2:B} and
      $\Sys_2$, we get $w'_0 = (t_1t_2)^\alpha t_1$ for some
      $\alpha\geq 0$. Therefore:
      $$
      \begin{array}{lllllll}
        & & v_0v_1\dots v_nv_mv_{\ol n}\dots v_{\ol 1}v_{\ol 0}
        \\ & = &
        (t_1t_2)^{\alpha_1+\dots+\alpha_n}v'_0 \\
        & = & (t_1t_2)^{\alpha_1+\dots+\alpha_n}w'_0w_{\ol 0}\text{ by
          $\Sys_1$} \\
        & = & (t_1t_2)^{\alpha_1+\dots+\alpha_n+\alpha}t_1w_{\ol 0} \\
        & = & w'_0(t_2t_1)^{\alpha_1+\dots+\alpha_n}w_{\ol 0} \\
        & = & w_0w_m(t_2t_1)^{\alpha_1+\dots+\alpha_n}w_{\ol 0} \\
        & = & w_0w_1\dots w_nw_mw_{\ol n}\dots w_{\ol 1}w_{\ol 0} \\
      \end{array}
      $$

    \item for all $k\in\{1,\dots,n\}$, $v_{\ol k} = \epsilon$ and there exists $p\in\{1,\dots,n\}$ such that
      $w_p \neq \epsilon$. By Lemma \ref{lem:allform}, there exist
      $t_1,t_2,t_3,t_4\in\Sigma^*$ and $\alpha_0,beta_0$ and
      $\alpha_i,\beta_i,\beta_{\ol i}\geq 0$ for all
      $i\in\{1,\dots,n\}$ such that $t_1t_2$ is primitive and for all
      $i\in\{1,\dots,n\}$, $t_1t_2=t_3t_4$, $t_4t_3w_m = w_mt_2t_1$,
      $v_i=(t_1t_2)^{\alpha_i}$, $w_i = (t_4t_3)^{\beta_i}$ and
      $w_{\ol i} = (t_2t_1)^{\beta_{\ol i}}$. Moreover, if $w_0 =
      \epsilon$, then $v_0 = (t_4t_3)^{\alpha_0}t_4$, and if $v_0 =
      \epsilon$, then $w_0 = (t_3t_4)^{\beta_0}t_3$. By Proposition
      \ref{prop:lengtheq}, since $v_{\ol k} = \epsilon$ for all
      $k\in\{1,\dots,n\}$, we get $\alpha_k = \beta_k + \beta_{\ol
        k}$. As for the case given in the paper, we can suppose that
      $v_0 = (t_3t_4)^{\alpha_0} = (t_1t_2)^{\alpha_0}$ and $w_0
      =(t_3t_4)^{\beta_0}t_3$. Indeed, if $w_0 = \epsilon$,
      we simply replaced $v_0$ by  $t_3v_0$ and $w_0$ by $t_3w_0$. 
      Finally:
      $$
      \begin{array}{cllll}
        & & v_0v_1\dots v_nv_mv_{\ol n}\dots v_{\ol 0}  \\
        & = &
        (t_1t_2)^{\alpha_0+\dots+\alpha_n}v_mv_{\ol 0} \\
        & = &
        (t_1t_2)^{\alpha_1+\dots+\alpha_n}v_0v_mv_{\ol 0} \\
        & = &
        (t_1t_2)^{\alpha_1+\dots+\alpha_n}w_0w_mw_{\ol 0} \text{ by $\Sys_1$}\\
        & = &
        (t_3t_4)^{\alpha_1+\dots+\alpha_n+\beta_0}t_3w_mw_{\ol 0}\\
        & = &
        w_0(t_4t_3)^{\alpha_1+\dots+\alpha_n}w_mw_{\ol 0}\\
        & = &
        w_0(t_4t_3)^{\beta_1+\dots+\beta_n} (t_4t_3)^{\beta_{\ol
            1}+\dots+\beta_{\ol n}}w_mw_{\ol 0}\text{ since $\alpha_i
          = \beta_i+\beta_{\ol i}$}\\
        & = &
        w_0(t_4t_3)^{\beta_1+\dots+\beta_n} w_m (t_2t_1)^{\beta_{\ol
            1}+\dots+\beta_{\ol n}} w_{\ol 0}\text{ since $t_4t_3w_m =
          w_mt_2t_1$}\\
        & = &f
        w_0w_1\dots w_n w_m w_{\ol 1}\dots w_{\ol n} w_{\ol 0}
      \end{array}
      $$

    \item for all $k\in\{1,\dots,n\}$, $w_{k} = \epsilon$ and there exists $p\in\{1,\dots,n\}$ such that
      $v_{\ol p} \neq \epsilon$. This case is symmetric to case $2$.
\end{enumerate}
\end{proof}

\subsection{Proof of Theorem \ref{thm:sys}: case $|v_\ell| = |w_\ell|$}

Remind that we have fixed some $n\geq 8$, some
words $u_m,v_m,w_m\in\Sigma^*$ and for all $i\in\{0,\dots,n\}$, 
we have fixed $u_i,v_i,w_i,u_{\ol i},v_{\ol i},w_{\ol i}\in \Sigma^*$ such that
$u_i,u_{\ol i}\neq\epsilon$ such that the following system holds:
$\Sys = \{ v_\pi = w_\pi\ |\ \pi:\{1,\dots,k\}\rightarrow
\{1,\dots,n\},\ |u_\pi|< |u_{id_n}|\}$.

Consider the following equations, defined for all
$a\in\mathbb{N}$, for all $i,k\in\{1,\dots,n\}$:

$$
\left\{\begin{array}{rcllllllll}
v_0v_mv_{\ol 0}  & = & w_0w_mw_{\ol 0} & (1) \\
v_0(v_\ell)^av_m(v_{\ol \ell})^av_{\ol 0} & = & w_0(w_\ell)^aw_m(w_{\ol \ell})^aw_{\ol 0} & (2)\\
v_0v_i(v_\ell)^av_m(v_{\ol \ell})^av_{\ol i}v_{\ol 0} & = &
w_0w_i(w_\ell)^aw_m(w_{\ol \ell})^aw_{\ol i}w_{\ol 0}  & (3) \\
v_0v_iv_k(v_\ell)^av_m(v_{\ol \ell})^av_{\ol k}v_{\ol i}v_{\ol 0} & = &
w_0w_iw_k(w_\ell)^aw_m(w_{\ol \ell})^aw_{\ol k}w_{\ol i}w_{\ol 0}  & (4) \\
v_0 \dots v_{\ell-1} v_{\ell+1} \dots v_n v_m v_{\ol n}\dots v_{\ol{\ell-1}} v_{\ol{\ell+1}} \dots v_{\ol 0} & = & 
w_0 \dots w_{\ell-1} w_{\ell+1} \dots w_n w_m w_{\ol n}\dots
w_{\ol{\ell-1}} w_{\ol{\ell+1}} \dots w_{\ol 0} & (5)
\end{array}\right.
$$
 
As done for the case $|v_\ell| > |w_\ell|$, we denoty by $\Sys_k$
the set of equations of type $k$, $k=1,\dots,5$.
As for the equations given in the paper for the case $|v_\ell| >
|w_\ell|$, we can prove similarly the
following proposition:

\begin{proposition}
  For all $k=1,\dots,5$, $\Sys_k$ holds.
\end{proposition}

As for the case $|v_\ell| > |w_\ell|$, we have the following
proposition (which is in fact indepent from the cases
$|v_\ell|= |w_\ell|$ or not):

\begin{proposition}\label{prop:lengtheq2}
  For all $i\in\{1,\dots,n\}$, $|v_iv_{\ol i}| = |w_iw_{\ol i}|$.
\end{proposition}

\paragraph{Case study} There are four cases:

\begin{enumerate}

\item[$(i)$] $|v_\ell| = |w_\ell|=  0$ and
$|v_{\ol \ell}| = |w_{\ol \ell}| = 0$;
\item[$(ii)$] $|v_\ell| = |w_\ell|
\neq 0$ and $|v_{\ol \ell}| = |w_{\ol \ell}| = 0$;
\item[$(iii)$] $|v_{\ol \ell}| = |w_{\ol \ell}| \neq 0$ and $|v_\ell|
  = |w_\ell| \neq 0$;
\item[$(iv)$] $|v_{\ol \ell}| = |w_{\ol \ell}| \neq 0$ and $|v_\ell|
  = |w_\ell| = 0$;
\end{enumerate}

Cases $(iv)$ is syntactically the same as case $(ii)$ if we consider
the mirror of the equations. Therefore we consider only
case $(i)$,$(ii)$ and $(iii)$. For each of those three cases, 
we prove that $v_{id_n} = w_{id_n}$ (Theorem \ref{thm:sys}).

Similarly as the case $|v_\ell|>|w_\ell|$, we can assume
wlog that $v_0 = \epsilon$ or $w_0 = \epsilon$, and
$v_{\ol 0} = \epsilon$ or $w_{\ol 0} = \epsilon$, otherwise
we remove their common prefixes in the systems
$\Sys_1,\dots,\Sys_5$.

\subsubsection{Subcase $|v_\ell| = |w_\ell|= |v_{\ol \ell}| = |w_{\ol
    \ell}| = 0$}

\begin{lemma}
If $|v_\ell| = |w_\ell|=  0$ and
$|v_{\ol \ell}| = |w_{\ol \ell}| = 0$, then 
$v_{id_n} = w_{id_n}$.
\end{lemma}

\begin{proof}
It is an obvious consequence of $\Sys_5$.\qed
\end{proof}

\subsubsection{Subcase $|v_\ell| = |w_\ell| \neq 0$ and $|v_{\ol \ell}| = |w_{\ol \ell}|
  \neq 0$}

\begin{lemma}\label{lemma:formabove}
There exist $t_1,t_2\in\Sigma^*$ such that $t_1t_2$ is primitive and
$\alpha_0,\beta_0,\alpha_\ell,\beta_\ell\geq 0$ such that:

$$
v_\ell =  (t_1t_2)^{\alpha_\ell}\qquad w_\ell=
(t_2t_1)^{\beta_\ell}\qquad
w_0 = \epsilon \Rightarrow v_0 = (t_2t_1)^{\alpha_0}t_2\qquad
v_0 = \epsilon \Rightarrow w_0 = (t_1t_2)^{\beta_0}t_1
$$
\end{lemma}

\begin{proof}
  Remind that by hypothesis, $v_\ell\neq \epsilon$. Then $w_\ell\neq \epsilon$. By $\Sys_2$,
  there exists $a\geq 0$ such that $(v_\ell)^a$ and $(w_\ell)^a$ 
  have a common factor of length at least $|v_\ell| + |w_\ell|  -
  \text{gcd}(|v_\ell|,|w_\ell|)$. By the fundamental lemma, there
  exist $t_1,t_2\in\Sigma^*$ such that $t_1t_2$ is primitive,
  $v_\ell \in (t_1t_2)^+$ and $w_\ell \in (t_2t_1)^+$.
  We now infer the form of $v_0$ when $w_0 = \epsilon$ (the form
  of $w_0$ when $v_0=\epsilon$ can be obtained by symmetry).
  Wlog, we can assume that $t_1\neq \epsilon$. Indeed,
  since $v_\ell\neq\epsilon$, we have $t_1t_2\neq\epsilon$, so that
  if $t_1 = \epsilon$, then we take $t'_1 = t_2$ and $t'_2 = t_1 =
  \epsilon$, and we have $v_\ell \in (t'_1t'_2)^+$ and $w_\ell\in
  (t'_2t'_1)^+$. By $\Sys_2$, we get
  $v_0(t_1t_2)^\omega = (t_2t_1)^\omega$. By Lemma \ref{lemma:overlap}.\ref{lemma:overlap2:A},
  $v_0 = (t_2t_1)^{\alpha_0}t_2$ for some $\alpha_0\geq 0$.\qed
\end{proof}

Since by hypothesis we have $|v_{\ol \ell}| = |w_{\ol \ell}|\neq 0$, 
by considering the mirror of the equations, we can prove the following
corollary of Lemma \ref{lemma:formabove}:

\begin{corollary}\label{coro:formabove}
There exist $t_3,t_4\in\Sigma^*$ such that $t_3t_4$ is primitive and
$\alpha_{\ol 0},\beta_{\ol 0},\alpha_{\ol \ell},\beta_{\ol \ell}\geq 0$ such that:

$$
v_{\ol \ell} =  (t_3t_4)^{\alpha_{\ol \ell}}\qquad w_\ell=
(t_4t_3)^{\beta_{\ol \ell}}\qquad
w_{\ol 0} = \epsilon \Rightarrow v_{\ol 0} = (t_3t_4)^{\alpha_{\ol 0}}t_3\qquad
v_{\ol 0} = \epsilon \Rightarrow w_{\ol 0} = (t_4t_3)^{\beta_0}t_4
$$
\end{corollary}

Under certain conditions, we can characterize the form of 
$v_i$'s and $w_i$'s:

\begin{lemma}\label{lemma:uk_neq_vk}\
If there exists $1\leq k\leq n$ such that $|v_k| \neq |w_k|$ then
there exist $\alpha_1,\dots,\alpha_n,\beta_1,\dots,\beta_n\geq 0$ such that
for all $i\neq k$:
$$
v_i =  (t_1t_2)^{\alpha_i}\qquad w_i= (t_2t_1)^{\beta_i}\qquad
$$
\end{lemma}

\begin{proof}
 There are two cases: either $v_0=\epsilon$ or $w_0=\epsilon$. We
 consider the second case only, the first being symmetric. By Lemma
 \ref{lemma:formabove}, $v_0 = (t_2t_1)^{\alpha_0}t_2$ for some
 $\alpha_0\geq 0$ and $t_1,t_2\in\Sigma^*$ with $t_1t_2$ primitive.
 By $\Sys_3$ and $\Sys_4$, we have:

  $$
  \begin{array}{lllllllllllllllllll}
    (1) & v_0v_i(t_1t_2)^\omega & = & w_i(t_2t_1)^\omega & \quad & (2) &
    v_0v_k(t_1t_2)^\omega & = & w_k(t_2t_1)^\omega & \quad & (3) &
    v_0v_kv_i(t_1t_2)^\omega & = & w_kw_i(t_2t_1)^\omega
  \end{array}
  $$

  \noindent  
  We again consider two cases:
  \begin{enumerate}
    \item $v_0v_k = w_kw$ for some $w$. $\Sys_2$
      gives $w(t_1t_2)^\omega = (t_2t_1)^\omega$. By Lemma
      \ref{lemma:overlap}.\ref{lemma:overlap2:A}, $w = (t_2t_1)^\beta t_2$ for some
      $\beta\geq 0$. $\Sys_3$ gives $wv_i(t_1t_2)^\omega =
      w_i(t_2t_1)^\omega$, and by $\Sys_1$, we get
      $wv_i(t_1t_2)^\omega = v_0v_i(t_1t_2)^\omega$, i.e.
      $(t_2t_1)^\beta t_2 v_i (t_1t_2)^\omega = (t_2t_1)^{\alpha_0}
      t_2v_i(t_1t_2)^\omega$. Since $|v_k|\neq |w_k|$ and 
      $v_0v_k = w_kw$, $|v_0|\neq |w|$, and 
      $\beta\neq \alpha_0$. Thus by taking
      $\gamma = |\alpha_0-\beta|>0$, we get
      $(t_1t_2)^{\gamma}v_i(t_1t_2)^\omega = v_i(t_1t_2)^\omega$.
      By Lemma \ref{lemma:overlap}.\ref{lemma:overlap2:G}, $v_i\in (t_1t_2)^*$.

    \item $w_k = v_0v_kv$ for some $v\neq \epsilon$. $\Sys_2$
      gives $(t_1t_2)^\omega = v(t_2t_1)^\omega$,
      i.e. $(t_1t_2)^\omega = vt_2(t_1t_2)^\omega$. Therefore by Lemma
      \ref{lemma:overlap}.\ref{lemma:overlap2:D}, $vt_2\in (t_1t_2)^\omega$. Since
      $v\neq\epsilon$, we get $v =  (t_1t_2)^\eta t_1$ for some
      $\eta\geq 0$. Now, $\Sys_3$ gives $v_i(t_1t_2)^\omega = vw_i(t_2t_1)^\omega$,
      and by $\Sys_1$, $v_i(t_1t_2)^\omega =
      vv_0v_i(t_1t_2)^\omega = (t_1t_2)^\eta t_1(t_2t_1)^{\alpha_0}
      t_2 v_i (t_1t_2)^\omega = (t_1t_2)^{\eta+\alpha_0+1}
      v_i(t_1t_2)^\omega$. By Lemma \ref{lemma:overlap}.\ref{lemma:overlap2:G}, $v_i\in
      (t_1t_2)^*$.
  \end{enumerate}
  
  \noindent In both cases $v_i\in (t_1t_2)^*$. By $\Sys_1$ 
  $v_0v_i(t_1t_2)^\omega = (t_2t_1)^\omega = w_i(t_2t_1)^\omega$ and
  by Lemma \ref{lemma:overlap}.\ref{lemma:overlap2:D} $w_i\in (t_2t_1)^*$.
\end{proof}

Again by considering the mirror of the equations, we can prove the
following corollary of Lemma \ref{lemma:uk_neq_vk}:

\begin{corollary}\label{coro:formeenforme}
If there exists $1\leq k\leq n$ such that $|v_{\ol k}| \neq |w_{\ol k}|$ then
there exist $\alpha_{\ol 1},\dots,\alpha_{\ol n},\beta_{\ol
  1},\dots,\beta_{\ol n}\geq 0$ such that
for all $i\neq k$:
$$
v_{\ol i} =  (t_3t_4)^{\alpha_{\ol i}}\qquad w_{\ol i}=
(t_4t_3)^{\beta_{\ol i}}\qquad
$$
\end{corollary}

\begin{lemma}\label{lemma:concatenation}
Let $\alpha\in\mathbb{N}$. If for all $i\in \{1\dots n\}$,
$|v_i|=|w_i|$ and there exist $a_i, b_i\in\mathbb{N}$ such that:
\begin{eqnarray}
(t_2t_1)^\alpha t_2 v_i (t_1t_2)^{a_i} = w_i (t_2t_1)^{b_i}t_2\label{eqn:toto}
\end{eqnarray}

then $$(t_2t_1)^\alpha t_2 v_1\dots v_n  = w_1 \dots w_n (t_2t_1)^{\alpha}t_2$$
\end{lemma}

\begin{proof}
From Eq.\ref{eqn:toto}, and $|v_i|=|w_i|$ we deduce that
$b_i=\alpha+a_i$, so that:

\begin{eqnarray}
(t_2 t_1)^{\alpha}t_2  v_i  = w_i (t_2 t_1)^{\alpha} t_2 \label{lemma:concatenation:eq1}
\end{eqnarray}

By induction on $n$ we show that $(t_2 t_1)^{\alpha}t_2 v_1\dots v_n= w_1\dots w_n (t_2 t_1)^{\alpha}t_2$.
Indeed, it is trivial if $n=0$. So suppose it is true for $n-1$, we have:
$$
\begin{array}{lllll}
&& (t_2 t_1)^{\alpha}t_2 v_1\dots v_{n} &\\
& = & w_1\dots w_{n-1} (t_2 t_1)^{\alpha}t_2 v_n&  \text{(by induction hypothesis)}\\
& = & w_1\dots w_{n} (t_2 t_1)^{\alpha}t_2& \text{(by (\ref{lemma:concatenation:eq1}))}
\end{array}
$$
\qed
\end{proof}

\begin{proposition}\label{prop:cases}
  One of the following propositions holds:
  
  \begin{enumerate}
    \item $\forall i\in\{1,\dots,n\}: v_i =  (t_1t_2)^{\alpha_i}\wedge
      w_i= (t_2t_1)^{\beta_i}\wedge v_{\ol i} = (t_3t_4)^{\alpha_{\ol
          i}}\wedge w_{\ol i} = (t_4t_3)^{\beta_{\ol i}}$
    \item $\exists k\in\{1,\dots,n\} \forall i\neq k: |v_i|=|w_i|$ and $|v_{\ol i}| =
      |w_{\ol i}|$
  \end{enumerate}
\end{proposition}

\begin{proof}
Indeed, if there are $k\neq k'$ such that
$|v_k|\neq|w_k|$ and $|v_{k'}|\neq|w_{k'}|$, then by Lemma
$\ref{lemma:uk_neq_vk}$ $\forall i: v_i =  (t_1t_2)^{\alpha_i}\wedge
w_i= (t_2t_1) ^{\beta_i}$. By Proposition \ref{prop:lengtheq2}, 
$|v_{\ol k}|\neq |w_{\ol k}|$ and $|v_{\ol k'}|\neq |w_{\ol k'}|$ so
that by Corollary \ref{coro:formeenforme}, for all $i$, 
$v_{\ol i} = (t_3t_4)^{\alpha_{\ol i}}$ and $w_{\ol i} =
(t_4t_3)^{\beta_{\ol i}}$.

Otherwise we have at most one $k$ with
$|v_k|\neq|w_k|$, and for all $i\neq k$, $|v_i| = |w_i|$, and by
Prop. \ref{prop:lengtheq2}, $|v_{\ol i}| = |w_{\ol i}|$.\qed
\end{proof}

We now prove Theorem \ref{thm:sys} for each of the cases
of Prop. \ref{prop:cases}. This is done in two lemmas: Lemma
\ref{lem:case1} and Lemma \ref{lem:case2}.

\begin{lemma}\label{lem:case1}
If for all $i\in\{1,\dots,n\}$, $v_i =  (t_1t_2)^{\alpha_i}$,
$w_i= (t_2t_1)^{\beta_i}$, $v_{\ol i} = (t_3t_4)^{\alpha_{\ol i}}$ and
$w_{\ol i} = (t_4t_3)^{\beta_{\ol i}}$, then 
$v_0\dots v_n v_m v_{\ol{n}}\dots v_{\ol
  0} = w_0 \dots w_{n} w_m w_{\ol{n}}\dots w_{\ol{0}}$.
\end{lemma}

\begin{proof}
 First by Lemma \ref{lemma:formabove} and Corollary
 \ref{coro:formabove}, we have:

$$
\begin{array}{lllllll}
w_0 = \epsilon & \Rightarrow & v_0 \in (t_2t_1)^*t_2 &\qquad & 
v_0 = \epsilon & \Rightarrow & w_0 \in (t_1t_2)^*t_1 \\
w_{\ol 0} = \epsilon & \Rightarrow & v_{\ol 0} \in (t_3t_4)^*t_3&\qquad&
v_{\ol 0} = \epsilon & \Rightarrow & w_{\ol 0} \in (t_4t_3)^*t_4
\end{array}
$$

Since $v_0=\epsilon$ or $w_0=\epsilon$, and $v_{\ol 0} = \epsilon$ or
$w_{\ol 0} = \epsilon$, we can assume wlog that
$v_0 = (t_1t_2)^{\alpha_0}$ and $w_0 = (t_1t_2)^{\beta_0}t_1$ for some
$\alpha_0,\beta_0\geq 0$. Indeed, if $w_0 = \epsilon$, we
simply replace in $\Sys_1,\dots,\Sys_5$ $v_0$ by $t_1v_0$ and
$w_0$ by $t_1w_0$ (which is indeed of the form $(t_1t_2)^*t_1$).
If $v_0=\epsilon$, then it is of the form $(t_1t_2)^*$ and 
$w_0$ is of the form $(t_1t_2)^*t_1$.

Similarly, we can assume wlog that
$v_{\ol 0} = (t_3t_4)^{\alpha_{\ol 0}}$ and $w_{\ol 0} =
(t_4t_3)^{\beta_{\ol 0}}t_4$ for some $\alpha_{\ol 0},\beta_{\ol
  0}\geq 0$.

Now, by $\Sys_1$ and $\Sys_2$, we have:
\begin{eqnarray*}
  v_0v_\ell v_m  v_{\ol{\ell}} v_{\ol 0} &= & w_0 w_\ell w_m w_{\ol{\ell}}w_{\ol{0}}\\
  v_0 v_m v_{\ol 0}   &= & w_0 w_m w_{\ol{0}}\\
\end{eqnarray*}

So we deduce:
 \begin{eqnarray*}
v_0v_\ell v_m  v_{\ol{\ell}} v_{\ol 0}&= & w_0 w_\ell w_m w_{\ol{\ell}}w_{\ol{0}}\\
\Leftrightarrow  (t_1 t_2)^{\alpha_0+\alpha_\ell}  v_m
(t_3t_4)^{\alpha_{\ol \ell}+\alpha_{\ol 0}}&= & w_0 w_\ell w_m w_{\ol{\ell}}w_{\ol{0}}\\
\Leftrightarrow  (t_1t_2)^{\alpha_\ell}(t_1t_2)^{\alpha_0}v_m
(t_3t_4)^{\alpha_{\ol 0}} (t_3t_4)^{\alpha_{\ol \ell}} &= & w_0w_\ell w_m w_{\ol{\ell}}w_{\ol{0}}\\
\Leftrightarrow  (t_1t_2)^{\alpha_\ell}  v_0  v_m v_{\ol 0}
(t_3t_4)^{\alpha_{\ol \ell}}&= & w_0 w_\ell w_m w_{\ol{\ell}}w_{\ol{0}}\\
\Leftrightarrow  (t_1t_2)^{\alpha_\ell}  w_0  w_m w_{\ol 0}
(t_3t_4)^{\alpha_{\ol \ell}} &= & w_0 w_\ell w_m w_{\ol{\ell}}w_{\ol{0}}\\
\Leftrightarrow  (t_1t_2)^{\alpha_\ell+\beta_0}t_1 w_m t_4
(t_3t_4)^{\alpha_{\ol \ell}+\beta_{\ol 0}} &= & w_0 w_\ell w_m w_{\ol{\ell}}w_{\ol{0}}\\
\Leftrightarrow  w_0(t_1t_2)^{\alpha_\ell} w_m 
(t_3t_4)^{\alpha_{\ol \ell}}w_{\ol 0} &= & w_0 w_\ell w_m w_{\ol{\ell}}w_{\ol{0}}\\
\Leftrightarrow  (t_1t_2)^{\alpha_\ell} w_m 
(t_3t_4)^{\alpha_{\ol \ell}} &= &  w_\ell w_m w_{\ol{\ell}}\\
\Leftrightarrow  (t_1t_2)^{\alpha_\ell}  w_m  (t_3t_4)^{\alpha_{\ol \ell}}&=
& (t_1t_2)^{\beta_\ell} w_m (t_3t_4)^{\beta_{\ol \ell}}
 \end{eqnarray*}

Then we conclude with:
\begin{eqnarray*}
&&v_0v_1\dots v_m v_{\ol{n}}\dots v_{\ol{1}}v_{\ol 0}\\
 &= &  (t_1 t_2)^{\alpha_0+\dots+\alpha_n}  v_m  (t_3t_4)^{\alpha_{\ol
      n}+\dots+\alpha_{\ol 0}}\\
 &= &  (t_1 t_2)^{\alpha_\ell}v_0\dots v_{\ell-1}v_{\ell+1}\dots
  v_nv_mv_{\ol n}\dots v_{\ol {\ell+1}}v_{\ol {\ell-1}}  (t_3t_4)^{\alpha_{\ol \ell}}\\
 &= &  (t_1 t_2)^{\alpha_\ell}w_0\dots w_{\ell-1}w_{\ell+1}\dots
  w_nw_mw_{\ol n}\dots w_{\ol {\ell+1}}w_{\ol {\ell-1}}
  (t_3t_4)^{\alpha_{\ol \ell}}\text{ by $\Sys_5$}\\
 &= &  (t_1
  t_2)^{\alpha_\ell+\beta_0+\dots+\beta_{\ell-1}+\beta_{\ell+1}\dots
    \beta_n}w_m(t_3t_4)^{\beta_{\ol n}+\dots\beta_{\ol
      {\ell+1}}+\beta_{\ol {\ell-1}}+\dots +\beta_{\ol 1}+\alpha_{\ol \ell}}\\
 &= &  (t_1
  t_2)^{\beta_0+\dots+\beta_{\ell-1}+\beta_{\ell+1}\dots
    \beta_n} (t_1t_2)^{\alpha_\ell}w_m(t_3t_4)^{\alpha_{\ol \ell}}(t_3t_4)^{\beta_{\ol n}+\dots\beta_{\ol
      {\ell+1}}+\beta_{\ol {\ell-1}}+\dots +\beta_{\ol 1}}\\
 &= &  (t_1
  t_2)^{\beta_0+\dots+\beta_{\ell-1}+\beta_{\ell+1}\dots
    \beta_n} (t_1t_2)^{\beta_\ell}w_m(t_3t_4)^{\beta_{\ol \ell}}(t_3t_4)^{\beta_{\ol n}+\dots\beta_{\ol
      {\ell+1}}+\beta_{\ol {\ell-1}}+\dots +\beta_{\ol 1}}\\
 &= &w_0w_1 \dots w_{n} w_m w_{\ol{n}}\dots w_{\ol{0}}
 \end{eqnarray*}\qed
\end{proof}

\begin{lemma}\label{lem:case2}
   If there exists $\exists k\in\{1,\dots,n\}$ such that for all
   $i\neq k$,  $|v_i|=|w_i|$ and $|v_{\ol i}| = |w_{\ol i}|$, then
   $v_0\dots v_n v_m v_{\ol{n}}\dots v_{\ol 0} = w_0 \dots w_{n} w_m
   w_{\ol{n}}\dots w_{\ol{0}}$.
\end{lemma}

\begin{proof}
By hypothesis, we have assumed that
$v_0=\epsilon$ or $w_0=\epsilon$, and $v_{\ol 0} = \epsilon$ or
$w_{\ol 0}=\epsilon$. This leads to four cases:

\begin{enumerate}
  \item $w_0 = \epsilon$ and $v_{\ol 0} = \epsilon$;
  \item $v_0 = \epsilon$ and $v_{\ol 0} = \epsilon$;
  \item $v_0 = \epsilon$ and $w_{\ol 0} = \epsilon$;
  \item $w_0 = \epsilon$ and $w_{\ol 0} = \epsilon$.
\end{enumerate}

We have assumed that $|v_\ell| = |w_\ell| \neq 0$ and
$|v_{\ol \ell}| = |w_{\ol \ell}| \neq 0$, and there is $k$ such that
for all $i\neq k$, $|v_i| = |w_i|$ and $|v_{\ol i}| = |w_{\ol i}|$.
This assumption is symmetric, so that with respect to the systems
$\Sys_1,\dots,\Sys_5$, cases $2$ and $4$ are symmetric, 
and case $1$ and $3$ are symmetric. Moreover, the proofs of cases $1$
and $2$ are very similar, therefore we focus on case $1$ only.

From now one, we assume that $w_0 = \epsilon$ and $v_{\ol 0} = \epsilon$.
By $\Sys_3$ and $v_\ell = (t_1t_2)^{\alpha_\ell}$ (Lemma \ref{lemma:formabove}) 
we have $v_0v_k(t_1t_2)^\omega = w_k(t_2t_1)^\omega$. Wlog we can
assume that $t_1\neq\epsilon$. Therefore by Lemma
\ref{lemma:overlap}.\ref{lemma:overlap2:A}, there exist $a_k,b_k$ such that
$v_0 v_k (t_1 t_2)^{a_k} = w_k (t_2 t_1)^{b_k} t_2$, equivalently we
consider two cases we suppose that either $a_k=0$ or that $a_k\neq 0, b_k=0$ i.e.
either $v_0 v_k  = w_k (t_2 t_1)^{b_k} t_2$ or $v_0 v_k (t_1 t_2)^{a_k-1} t_1 = w_k $. 

\begin{itemize}
\item Case $v_0 v_k  = w_k (t_2 t_1)^{b_k} t_2$:

	First, we know that $|v_i|=|w_i|$ for all $i< k$ and that
        there are $a_i,b_i\in\mathbb{N}$ with $v_0 v_i (t_1t_2)^{a_i}
        =  w_i (t_2t_1)^{b_i}t_2$
 (by $\Sys_3$ and Lemma \ref{lemma:overlap}.\ref{lemma:overlap2:A}) where $v_0=(t_1t_2)^{\alpha_0} t_2$, so by Lemma \ref{lemma:concatenation}  we have:
	\begin{eqnarray}
	 v_0 v_1\dots v_{k-1}= w_1\dots w_{k-1} v_0
	\end{eqnarray} 
	
	Second we have $v_0 v_k  = w_k (t_2 t_1)^{b_k} t_2$ by hypothesis (the case we are considering).
	
	Third, again by $\Sys_3$ and Lemma
        \ref{lemma:overlap}.\ref{lemma:overlap2:A} we know that $|v_i|=|w_i|$ for all $i> k$ and that there are $a'_i,b'_i\in\mathbb{N}$ with $v_0 v_k v_i (t_1t_2)^{a'_i} =  w_k w_i (t_2t_1)^{b'_i}t_2$ i.e. by replacing $v_0 v_k$ with $w_k (t_2 t_1)^{b_k} t_2$ we have $(t_2t_1)^{b_k}t_2 v_i (t_1t_2)^{a'_i} =  w_i (t_2t_1)^{b'_i}t_2$, so by Lemma \ref{lemma:concatenation} we have:
	\begin{eqnarray}
	 (t_2 t_1)^{b_k} t_2 v_{k+1}\dots v_n= w_{k+1}\dots w_n (t_2 t_1)^{a_k} t_2
	\end{eqnarray}

	As a consequence we have:

\begin{eqnarray}
&& v_0 \dots v_n \nonumber\\
 &= & v_0 \dots v_{k-1} v_{k} v_{k+1} \dots v_n \nonumber\\ 
 &= & w_1 \dots w_{k-1} v_0 v_{k} v_{k+1} \dots v_n \nonumber\\
 &= & w_1 \dots w_{k-1} w_{k} (t_2 t_1)^{b_k} t_2 v_{k+1} \dots v_n \nonumber\\
 & = & w_1\dots w_{k-1}w_{k} w_{k+1} \dots w_n (t_2 t_1)^{b_k} t_2 \nonumber\\
\label{lem:case2:formu0_un} & = & w_1 \dots w_n (t_2 t_1)^{b_k} t_2
 \end{eqnarray}
		
\item Case $v_0 v_k (t_1 t_2)^{a} t_1 = w_k $: We can show that $v_0 \dots v_n (t_2 t_1)^{a_k-1} t_2 = w_1 \dots w_n $ with a very similar proof.
\end{itemize}

By symmetry (since $v_{\ol 1}\neq \epsilon$ and $w_{\ol 1}\neq \epsilon$), we have either
	  $t_3 (t_4 t_3)^{d_k} v_{\ol {k}} = w_{\ol{k}} w_{\ol{0}}$ or $v_{\ol {k}} = t_4 (t_3 t_4)^{c_k}   w_{\ol{k}} w_{\ol{0}}$. 

We conclude the proof by putting this together and showing that $v_0v_1\dots v_m v_{\ol{n}}\dots v_{\ol{1}}=w_1 \dots w_{n} w_m w_{\ol{n}}\dots w_{\ol{0}}$:
	 
	\begin{itemize}
		\item Subcase $t_3 (t_4 t_3)^{d_k} v_{\ol {k}} = w_{\ol{k}} w_{\ol{0}}  $: this implies that  $t_3 (t_4 t_3)^{d_k} v_{\ol {n}}\dots v_{\ol 1}  = w_{\ol{n}}\dots w_{\ol{0}}  $. 
		Moreover we know that $v_0 v_k v_m v_{\ol k} = w_k w_m w_{\ol{k}} w_{\ol{0}}$ i.e. $(t_2 t_1)^{b_k} t_2 v_m  = w_m t_3 (t_4 t_3)^d$. We can deduce:
\begin{eqnarray*}
&&v_0v_1\dots v_m v_{\ol{n}}\dots v_{\ol{1}}\\
 &= &w_1 \dots w_{n} (t_1 t_2)^{b_k} t_1  v_m  v_{\ol{n}}\dots v_{\ol{1}}\\
 &= &w_1 \dots w_{n} w_m t_3 (t_4 t_3)^{d_k} v_{\ol{n}}\dots v_{\ol{1}}\\
 &= &w_1 \dots w_{n} w_m w_{\ol{n}}\dots w_{\ol{0}}
 \end{eqnarray*}

		\item Subcase $v_{\ol {k}} = t_4 (t_3 t_4)^c   w_{\ol{k}} w_{\ol{0}}  $: this implies that  $v_{\ol {n}}\dots v_{\ol 1}  = t_4 (t_3 t_4)^c w_{\ol{n}}\dots w_{\ol{0}}  $. 
		Moreover we know that $v_0 v_k v_m v_{\ol k} = w_k w_m w_{\ol{k}} w_{\ol{0}}$ i.e. $(t_2 t_1)^{b} t_2 v_m t_4 (t_3 t_4)^c = w_m $. We can deduce:
\begin{eqnarray*}
&&v_0v_1\dots v_m v_{\ol{n}}\dots v_{\ol{1}}\\
 &= &w_1 \dots w_{n} (t_1 t_2)^{b} t_1  v_m  v_{\ol{n}}\dots v_{\ol{1}}\\
 &= &w_1 \dots w_{n} (t_1 t_2)^{b} t_1  v_m  t_4 (t_3 t_4)^c w_{\ol{n}}\dots w_{\ol{0}}\\
 &= &w_1 \dots w_{n} w_m w_{\ol{n}}\dots w_{\ol{0}}
 \end{eqnarray*}
\end{itemize}
\qed	
\end{proof}


\subsubsection{Subcase $|v_\ell| = |w_\ell| \neq 0$ and $|v_{\ol \ell}| = |w_{\ol \ell}|
  = 0$}

Similarly as Proposition \ref{prop:cases}, one can prove the following
proposition:

\begin{proposition}\label{prop:cases2}
  One of the following propositions holds:
  
  \begin{enumerate}
    \item $\forall i\in\{1,\dots,n\}: v_i =  (t_1t_2)^{\alpha_i}\wedge
      w_i= (t_2t_1)^{\beta_i}$
    \item $\exists k\in\{1,\dots,n\} \forall i\neq k: |v_i|=|w_i|$.
  \end{enumerate}
\end{proposition}

\begin{lemma}
  If $|v_\ell| = |w_\ell| \neq 0$ and $|v_{\ol \ell}| = |w_{\ol \ell}|
  = 0$, then $v_{id_n} = w_{id_n}$.
\end{lemma}

\begin{proof}
Let pose $V_1=v_0\dots v_{\ell-1}v_{\ell+1}\dots v_n$, resp. $W_1=w_0\dots w_{\ell-1}w_{\ell+1}\dots w_n$, and 
$V=v_mv_{\ol n}\dots v_{{\ol \ell+1}}v_{\ol \ell-1}\dots v_{\ol 0}=v_mv_{\ol n}\dots v_{\ol 0}$, resp. 
$W=w_mw_{\ol n}\dots w_{{\ol \ell+1}}w_{\ol \ell-1}\dots w_{\ol 0}=w_mw_{\ol n}\dots w_{\ol 0}$.
By $\Sys_5$ we have $V_1V=W_1W$. 
We can suppose wlog that $W_1=V_1W'$, i.e. we have:
\begin{eqnarray}
V =  W'W \label{subcaseeq:eq1}
\end{eqnarray}

Now let $V_2=v_0\dots v_n$ and $W_2=w_0\dots w_n$. 
We have $v_{id_n}=V_2V$ and $w_{id_n}=W_2W$. We will show that $W_2=V_2W'$. 
This will conclude the proof as with Eq. \ref{subcaseeq:eq1} we have $v_{id_n}=V_2V= V_2W'W= W_2W =w_{id_n}$.

First note that Lemma \ref{lemma:formabove} is valid in this context and therefore we have 
$w_0 = \epsilon \Rightarrow  v_0 \in (t_2t_1)^*t_2$ and 
$v_0 = \epsilon  \Rightarrow  w_0 \in (t_1t_2)^*t_1$, 
as above we can consider that $v_0=(t_2t_1)^{\alpha_0}t_2 $ and $w_0=(t_2t_1)^{\beta_0}$.

We consider two cases following Proposition \ref{prop:cases2}:

\begin{enumerate}
\item $\forall i\in\{1,\dots,n\}: v_i =  (t_1t_2)^{\alpha_i}\wedge w_i= (t_2t_1)^{\beta_i}$: 
      Let write $\alpha= \alpha_0+\dots + \alpha{\ell-1}+ \alpha{\ell+1}+\dots+\alpha_n$ and $\beta_0+\dots + \beta{\ell-1}+ \beta{\ell+1}+\dots+\beta_n$
      we have $V_1=v_0\dots v_{\ell-1}v_{\ell+1}\dots v_n= (t_2t_1)^{\alpha} t_2$
      and $W_1=w_0\dots w_{\ell-1}w_{\ell+1}\dots w_n= (t_2t_1)^{\beta} $,
      therefore $W'=(t_2t_1)^{\alpha-\beta} t_2$. 
      Moreover  $V_2=V_1 (t_1t_2)^{\alpha_\ell}$ and $W_2=W_2  (t_1t_2)^{\alpha_\ell}$, as a result $W_2=V_2W'$. 
      
\item $\exists k\in\{1,\dots,n\} \forall i\neq k: |v_i|=|w_i|$:       
	By using the same construction as for Eq. \ref{lem:case2:formu0_un} of Lemma \ref{lem:case2}, 
	we can show that there exists $\alpha_k$ such that $W' = (t_2 t_1)^{a_k} t_2$ with $W_1=V_1W'$ and $W_2=V_2W'$.
\end{enumerate}
\qed

\end{proof}

%% file: figS2.tex
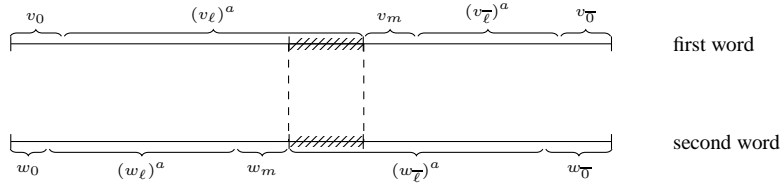
\begin{figure}[!ht]
\centering
\begin{tikzpicture}

\draw[|-|] (0,0) node[left=2mm] {} -- (8,0) node[right=7mm] 
  {\fontsize{8pt}{7pt}\selectfont first word}; 

\draw[snake=brace] (0,.1) -- (0.68,.1) node[midway,above=2pt] {$\scriptstyle v_0$};
\draw[snake=brace] (0.72,.1) -- (4.68,.1) node[midway,above=2pt] {$\scriptstyle (v_\ell)^a$};
\draw[snake=brace] (4.72,.1) -- (5.38,.1) node[midway,above=2pt] {$\scriptstyle v_m$};
\draw[snake=brace] (5.42,.1) -- (7.28,.1) node[midway,above=2pt] {$\scriptstyle (v_{\ol \ell})^a$};
\draw[snake=brace] (7.32,.1) -- (7.98,.1) node[midway,above=2pt] {$\scriptstyle v_{\ol 0}$};

\draw[|-|] (0,-1.3) node[left=2mm] {} -- (8,-1.3) node[right=7mm] 
  {\fontsize{8pt}{7pt}\selectfont second word}; 

\draw[snake=brace] (8,-1.4) -- (7.12,-1.4) node[midway,below=3pt] {$\scriptstyle w_{\ol 0}$};
\draw[snake=brace] (7.08,-1.4) -- (3.72,-1.4) node[midway,below=2pt] {$\scriptstyle (w_{\ol \ell})^a$};
\draw[snake=brace] (3.68,-1.4) -- (3.02,-1.4) node[midway,below=3pt] {$\scriptstyle w_m$};
\draw[snake=brace] (2.98,-1.4) -- (0.52,-1.4) node[midway,below=2pt] {$\scriptstyle (w_\ell)^a$};
\draw[snake=brace] (0.48,-1.4) -- (0,-1.4) node[midway,below=3pt] {$\scriptstyle w_0$};

\draw[|-|] (3.7,0) -- (4.7,0);
\draw[|-|] (3.7,-1.3) -- (4.7,-1.3);
\draw[dashed] (3.7,0) -- (3.7,-1.3);
\draw[dashed] (4.7,0) -- (4.7,-1.3);

\begin{scope}
\path[clip] (3.6,-.07) -| (4.7,.07) -| (3.6,-.07);
\foreach \i in {3.6,3.7,...,4.7} 
  {\draw (\i,-.2) -- +(.4,.4);}
\end{scope}

\begin{scope}[xshift=0mm,yshift=-1.3cm]
\path[clip] (3.6,-.07) -| (4.7,.07) -| (3.6,-.07);
\foreach \i in {3.6,3.7,...,4.7} 
  {\draw (\i,-.2) -- +(.4,.4);}
\end{scope}

\end{tikzpicture}
\caption{System $\Sys_2$ for large values of $a$, case $|v_\ell|> |w_\ell|$.}\label{figS2}
\end{figure}

%% file: figS5.tex
\begin{figure}[!ht]
\centering
\begin{tikzpicture}

\draw[dashed] (0,0) -- (1,0);
\draw[-|] (1,0) node[left=2mm] {} -- (8,0) node[right=3mm] 
  {\fontsize{8pt}{7pt}\selectfont end of the first word}; 

\draw[snake=brace] (0,.1) -- (1.98,.1) node[midway,above=2pt] {$\scriptstyle (v_\ell)^a$};
\draw[snake=brace] (2.02,.1) -- (2.48,.1) node[midway,above=2pt] {$\scriptstyle v_i$};
\draw[snake=brace] (2.52,.1) -- (4.28,.1) node[midway,above=2pt] {$\scriptstyle (v_\ell)^{b_0}$};
\draw[snake=brace] (4.32,.1) -- (4.78,.1) node[midway,above=2pt] {$\scriptstyle v_m$};
\draw[snake=brace] (4.82,.1) -- (5.48,.1) node[midway,above=2pt] {$\scriptstyle (v_{\ol \ell})^{b_0}$};
\draw[snake=brace] (5.52,.1) -- (5.98,.1) node[midway,above=2pt] {$\scriptstyle v_{\ol i}$};
\draw[snake=brace] (6.02,.1) -- (7.28,.1) node[midway,above=2pt] {$\scriptstyle (v_{\ol \ell})^a$};
\draw[snake=brace] (7.32,.1) -- (7.98,.1) node[midway,above=2pt] {$\scriptstyle v_{\ol 0}$};

\draw[dashed] (0,-1.3) -- (1,-1.3);
\draw[-|] (1,-1.3) node[left=2mm] {} -- (8,-1.3) node[right=3mm] 
  {\fontsize{8pt}{7pt}\selectfont end of the second word}; 

\draw[snake=brace] (8,-1.4) -- (6.52,-1.4) node[midway,below=3pt] {$\scriptstyle w_{\ol 0}$};
\draw[snake=brace] (6.48,-1.4) -- (1,-1.4) node[midway,below=2pt] {$\scriptstyle (w_{\ol \ell})^a$};

\draw[|-|] (1.5,0) -- (3,0);
\draw[|-|] (1.5,-1.3) -- (3,-1.3);
\draw[dashed] (1.5,0) -- (1.5,-1.3);
\draw[dashed] (3,0) -- (3,-1.3);

\begin{scope}
\path[clip] (1.4,-.07) -| (3,.07) -| (1.4,-.07);
\foreach \i in {1.4,1.5,...,3} 
  {\draw (\i,-.2) -- +(.4,.4);}
\end{scope}

\begin{scope}[xshift=0mm,yshift=-1.3cm]
\path[clip] (1.4,-.07) -| (3,.07) -| (1.4,-.07);
\foreach \i in {1.4,1.5,...,3} 
  {\draw (\i,-.2) -- +(.4,.4);}
\end{scope}


\draw[snake=brace] (2.98,-0.15) -- (1.52,-0.15) node[midway,below=2pt] {$\scriptstyle v_\ell v_iv_\ell$};


\draw[dashed] (8,0.7) -- (8,0);
\draw[dashed] (1.5,0.7) -- (1.5,0);

\draw[dashed] (8,-1.3) -- (8,-2);
\draw[dashed] (1,-1.3) -- (1,-2);

\draw[<->,dotted] (1.5,0.7) -- (8,0.7) node[midway,above=1pt] {$\scriptstyle l_1(a)$};
\draw[<->,dotted] (1,-2) -- (8,-2) node[midway,below=1pt] {$\scriptstyle l_2(a)$};

\end{tikzpicture}
\caption{System $\Sys_5$ for value $b_0$ and large values of $a$, case $|v_\ell|> |w_\ell|$.}\label{figS5}
\end{figure}
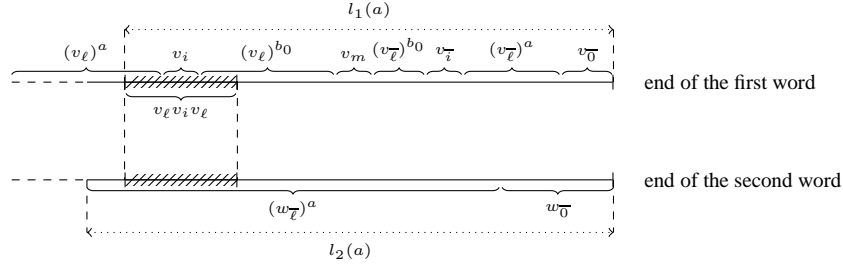

%% file: pspace.tex
\section{A \pspace algorithm for functionality}\label{sec:pspace}

We now show how the pumping lemma for functionality can be used to decide functionality 
in \pspace. It relies an \nlogspace algorithm for
functionality of \fsts, which is a consequence
of the following pumping argument by Sch{\"u}tzenberger:

\begin{theorem}[Sch{\"u}tzenberger, 1975 \cite{Schutz75}]
 Let $T$ be an \fst with $m$ states. If $T$ is non-functionnal then there
exists a word $w$ of length at most $3*m^2$ that admits two different
outputs.
\end{theorem}

As a consequence, we obtain:

\begin{theorem}\label{thm:funfst}
  Functionality of \fsts is decidable in \nlogspace.
\end{theorem}

\begin{proof}
We give a \conlogspace algorithm. The result follows as
\conlogspace = \nlogspace.

Note that each transition outputs a sequence of letters of bounded
length, therefore one can bound polynomially the length of the two different
outputs for a single input that witnesses non-functionality.
Let us point out that two outputs differ either because one is a
strict prefix of the other, or on a common position their letters
differ. By a small trick and a new dummy symbol in the input alphabet,
it is easy to reduce the first case to the second one with an
augmentation of the FST of constant size.

We consider a non-deterministic algorithm for deciding
non-functionality, operating as follows: one guesses a position $i$ in
the output where two outputs differ. Then using only logarithmic space,
one can check that this guess is correct. At each step, this algorithm guesses
itself one letter of the input and the two transitions
of the two runs computing the two different outputs. Therefore at each
step, this algorithm keeps two counters and the two states reached by the two
runs so far. The first (resp. second) counter counts the length of the
first (resp. second) output. When one of the outputs has reached
position $i$, the algorithm stores the $i$-th letter of
this output, and continue until the other output reaches the
$i$-th position. At this point, the two runs are in two states $p,q$,
and one just has to check whether the two
letters at the $i$-th position are different. Finally, the algorithm checks whether the two
runs can be continued into successful runs (from $p$ and $q$) on the
same input. This can be again done in non-deterministic logarithmic space.

By Sch{\"u}tzenberger's Theorem, one can take $i\leq 3m^2$, and
therefore the two counters are represented in logarithmic space in the
size of the \fst.\qed
\end{proof}

We can now give a \pspace algorithm for functionality. We devise a construction which
given a \vpt $A$, builds an \fst $B$ that simulates $A$ for nested input
words of small height. The height of the input word being polynomially
bounded (Lemma \ref{lem:pumpvert}), one can bound similarly the height
of the stack of the \vpt. Then, as runs cross only finitely many stacks, one can incorporate these stacks into a
finite-control part, turning the \vpt into an \fst. This construction is correct in the following sense:

\begin{proposition}
For all \vpt $A$ with $n$ states, one can construct an \fst $B$ of exponential
size wrt $n$, such that $\dom(B) = \{ u\in\dom(A)\ |\ h(u)\leq 8n^4\}$
and for all $w\in\dom(B)$, $B(w) = A(w)$. Moreover, $A$ is functional
iff $B$ is functional.
\end{proposition}

The idea is to apply the \nlogspace algorithm of Theorem \ref{thm:funfst} on $B$.
However, building this \fst $B$ of exponential size wrt to the size of
the \vpt $A$ as the first step of an algorithm will not yield a $\textsc{Pspace}$
algorithm. Therefore, the construction of the transition rules of 
$B$ has to be performed on-demand when such a transition is needed.
Altogether, this gives a $\textsc{Pspace}$ algorithm for deciding
functionality of \vpts.